\PassOptionsToPackage{usenames,dvipsnames}{xcolor}

\documentclass[a4paper,11pt]{article}

\usepackage{todonotes,soul,microtype}
\usepackage{amssymb,amsmath,amsthm}

\usepackage{thm-restate}

\usepackage[margin=1in]{geometry}

\usepackage{cite}
\usepackage[bookmarks=false,colorlinks,breaklinks]{hyperref}
\hypersetup{linkcolor=blue,citecolor=blue,filecolor=blue,urlcolor=blue}

\allowdisplaybreaks

\newtheorem{theorem}{Theorem}[section]
\newtheorem{lemma}{Lemma}[section]
\newtheorem{corollary}{Corollary}[section]
\newtheorem{definition}{Definition}[section]

\newtheorem{observation}{Observation}
\newtheorem{proposition}{Proposition}
\newtheorem{rr}{Reduction Rule}

\newcommand{\C}{\mathcal}
\newcommand{\OO}{{\mathcal O}}
\newcommand{\N}{{\mathbb{N}}}
\newcommand{\Nb}{{\mathbb{N}}}
\newcommand{\I}{\ensuremath{\mathcal{I}}}
\newcommand{\J}{\ensuremath{\mathcal{J}}}
\newcommand{\A}{\ensuremath{\mathcal{A}}}
\newcommand{\F}{\ensuremath{\mathcal{F}}}
\newcommand{\To}{\ensuremath{\rightarrow}}
\newcommand{\Sb}[1]{\ensuremath{\left\{#1\right\}}}
\newcommand{\restrict}[1]{\raisebox{0ex}{$| $}_{#1}}
\newcommand{\sm}{\setminus}
\newcommand{\Tb}[1]{\ensuremath{\left[#1\right]}}

\newcommand{\Mopt}{\ensuremath{\mu_{M}}}
\newcommand{\Wopt}{\ensuremath{\mu_{W}}}
\newcommand{\bal}[1]{\ensuremath{{\sf balance}(#1)}}
\newcommand{\mbal}{\ensuremath{{\sf balance}}}
\newcommand{\worst}{\ensuremath{{\sf worst}}}
\newcommand{\Bal}{\ensuremath{{\sf Bal}}}
\newcommand{\Om}{\ensuremath{O_{M}}}
\newcommand{\Ow}{\ensuremath{O_{W}}}

\newcommand{\NPH}{\textrm{\textup{NP-hard}}}
\newcommand{\Yes}{{\sc Yes}}
\newcommand{\No}{{\sc No}}
\newcommand{\FPT}{\textrm{\textup{FPT}}}
\newcommand{\WOH}{\textrm{\textup{W[1]-hard}}}

\newcommand{\minBSM}{{\sc Above-Min BSM}}
\newcommand{\pBSM}{\minBSM} 
\newcommand{\maxBSM}{{\sc Above-Max BSM}}
\newcommand{\FBSM}{{\sc Above-Min FBSM}} 
\newcommand{\pGBSM}{\FBSM} 
\newcommand{\sad}{sad}
\newcommand{\happy}{happy}
\newcommand{\bad}{\sad} 
\newcommand{\good}{\happy} 
\newcommand{\Msad}{\ensuremath{M_{S}}} 
\newcommand{\Mhap}{\ensuremath{M_H}} 
\newcommand{\Wsad}{\ensuremath{W_{S}}} 
\newcommand{\Whap}{\ensuremath{W_H}}
\newcommand{\Mbad}{\Msad} 
\newcommand{\Mgood}{\Mhap} 
\newcommand{\Wbad}{\Wsad} 
\newcommand{\Wgood}{\Whap} 

\newcommand{\defparproblemMy}[4]{
  \vspace{1mm}
  \noindent\fbox{
  \begin{minipage}{0.96\textwidth}
  \begin{tabular*}{\textwidth}{@{\extracolsep{\fill}}lr} #1\\ \end{tabular*}
  {\bf Input:} #2  \\
  {\bf Question:} #4\\
  {\bf Parameter:} #3 
  \end{minipage}}
  \vspace{1mm}}

\usepackage{comment}
\excludecomment{previous}

\newcommand{\footremember}[2]{%
    \footnote{#2}
    \newcounter{#1}
    \setcounter{#1}{\value{footnote}}%
}
\newcommand{\footrecall}[1]{%
    \footnotemark[\value{#1}]%
}

\begin{document}

\title{Balanced Stable Marriage: How Close is Close Enough?\footnote{The research leading to these results received funding from the European Research Council under the European Union’s Seventh Framework Programme (FP/2007-2013) / ERC Grant Agreement no.~306992.}}

\author{Sushmita Gupta\footremember{alley}{University of Bergen, Norway; Emails: {\tt \{Sushmita.Gupta, Meirav.Zehavi\}@uib.no}} \and Sanjukta Roy\footremember{trailer}{Institute of Mathematical Sciences, India; Emails: {\tt \{sanjukta, saket\}@imsc.res.in} } \and Saket Saurabh\footrecall{alley}\, \footrecall{trailer} \and Meirav Zehavi\footrecall{alley}}

\date{ }

\maketitle
\thispagestyle{empty}

\begin{abstract}
The {\sc Balanced Stable Marriage} problem is a central optimization version of the classic {\sc Stable Marriage} problem. Here, the output cannot be an arbitrary stable matching, but one that balances between the dissatisfaction of the two parties, men and women. We study {\sc Balanced Stable Marriage} from the viewpoint of Parameterized Complexity. Our ``above guarantee parameterizations''   are arguably the most natural parameterizations of the problem at hand. Indeed, our parameterizations precisely fit the scenario where there exists a stable marriage that both parties would accept, that is, where the satisfaction of each party is {\em close} to the best it can hope for. Furthermore, our parameterizations accurately draw the line between tractability and intractability with respect to the target value.
\end{abstract}



\section{Introduction}\label{sec:intro}

{\em Matching under preferences} is a rich topic central to both economics and computer science, which has been consistently and intensively studied for over several decades. One of the main reasons for interest in this topic stems from the observation that it is extremely relevant to a wide variety of practical applications modeling situations where the objective is to {\em match} agents to other agents (or to resources). In the most general setting, a matching is defined as an allocation (or assignment) of agents to resources that satisfies some predefined criterion of compatibility/acceptability. Here, the (arguably) best known model is the {\em two-sided model}, where the agents on one side are referred to as {\it men}, and the agents on the other side are referred to as {\it women}. A few illustrative examples of real life situations where this model is employed in practice include matching hospitals to residents, students to colleges, kidney patients to donors and users to servers in a distributed Internet service. At the heart of all of these applications lies the fundamental {\sc Stable Marriage} problem. In particular, the Nobel Prize in Economics was awarded to Shapley and Roth in 2012 ``for the theory of stable allocations and the practice of market design.'' Moreover, several books have been dedicated to the study of {\sc Stable Marriage} as well as optimization variants of this classical problem such as the {\sc Egalitarian Stable Marriage}, {\sc Sex-Equal Stable Marriage} and {\sc Balanced Stable Marriage} problems \cite{DBLP:books/daglib/0066875,opac-b1092346,DBLP:books/ws/Manlove13}.

The input of the {\sc Stable Marriage} problem consists of a set of men, $M$, and a set of women, $W$, each person ranking a subset of people of the opposite gender. That is, each person $a$ has a set of {\em acceptable partners}, ${\cal A}(a)$, whom this person subjectively ranks. Consequently, each person $a$ has a so-called {\em preference list}, where $p_{a}(b)$ denotes the position of $b\in{\cal A}(a)$ in $a$'s preference list. Without loss of generality, it is assumed that if a person $a$ ranks a person $b$, then the person $b$ ranks the person $a$ as well. The sets of preference lists of the men and the women are denoted by ${\cal L}_M$ and ${\cal L}_W$, respectively. In this context, we say that a pair of a man and a woman, $(m,a)$, is an {\em acceptable pair} if both $m\in{\cal A}(w)$ and $w\in{\cal A}(m)$. Accordingly, the notion of a {\em matching} refers to a matching between men and women, where two people that are matched to one another form an acceptable pair. Roughly speaking, the goal of the {\sc Stable Marriage} problem is to {\em find} a matching that is {\em stable} in the following sense: there should not exist two people who prefer being matched to one another over their current ``status''. More precisely, a matching $\mu$ is said to be stable if it does not have a {\em blocking pair}, which is an acceptable pair $(m,w)$ such that {\bf (i)} either $m$ is unmatched by $\mu$ or $p_{m}(w)< p_{m}(\mu(m))$, and {\bf (ii)} either $w$ is unmatched by $\mu$ or  $p_{w}(m)< p_{w}(\mu(w))$. Here, the notation $\mu(a)$ indicates the person to whom $\mu$ matches the person $a$. Note that a person always prefers being matched to an acceptable partner over being unmatched.

The seminal paper \cite{gale62a} by Gale and Shapely on stable matchings shows that given an instance of {\sc Stable Marriage}, a stable matching necessarily {\em exists}, but it is not necessarily unique. In fact, for a given instance of {\sc Stable Marriage}, there can be an {\em exponential} number of stable matchings, and they should be viewed as a {\em spectrum} where the two extremes are known as the {\em man-optimal stable matching} and the {\em woman-optimal stable matching}. Formally, the  man-optimal stable matching, denoted by \Mopt, is a stable matching such that every stable matching $\mu$ satisfies the following condition: every man $m$ is either unmatched by both \Mopt\ and $\mu$ or $p_{m}(\Mopt{}(m))\leq p_{m}(\mu(m))$. The woman-optimal stable matching, denoted by \Wopt, is defined analogously. These two extremes, which give the best possible solution for one party at the expense of the other party, always exist and can be computed in polynomial time \cite{gale62a}. Naturally, it is desirable to analyze matchings that lie somewhere in the middle, being {\em globally desirable}, {\em fair towards both sides} or {\em desirable by both sides}. 

Each notion above of what constitutes a desirable stable matching leads to a natural, {\em different} optimization problem. The determination of which notion best describes an appropriate outcome depends on the specific situation at hand. Here, the quantity $p_{a}(\mu(a))$ is viewed as the ``satisfaction'' of $a$ in a matching  $\mu$, where a smaller value signifies a greater amount of  satisfaction. Under this interpretation, the {\it egalitarian stable matching} attempts to be {\em globally desirable} by minimizing $e(\mu)=\sum_{(m,w)\in \mu} (p_{m}(\mu(m))+p_{w}(\mu(w)))$ over the set of all stable matchings, which we denote by $\mathrm{SM}$. The problem of finding an egalitarian stable matching, called {\sc Egalitarian Stable Marriage}, is known be solvable in polynomial time due to Irving et al.~\cite{IrvingLG87j}. Roughly speaking, this problem does not distinguish between men and women, and therefore it does not fit scenarios where it is necessary to differentiate between the individual satisfaction of each party. In such scenarios, the {\sc Sex-Equal Stable Marriage} and {\sc Balanced Stable Marriage} problems come into play. Before we define each of these two problems, we would like to remark that a survey of results related to {\sc Egalitarian Stable Marriage} and {\sc Sex-Equal Stable Marriage} is outside the scope of this paper, and we refer interested readers to the books \cite{DBLP:books/daglib/0066875,opac-b1092346,DBLP:books/ws/Manlove13}. Here, we consider these two problems only to understand the context of {\sc Balanced Stable Marriage}.

In the {\sc Sex-Equal Stable Marriage} problem, the objective is to find a stable matching that minimizes the absolute value of $\delta(\mu)$ over $\mathrm{SM}$, where $\delta(\mu)=\sum_{(m,w)\in \mu} p_{m}(\mu(m))$ $-\sum_{(m,w)\in \mu}p_{w}(\mu(w))$. It is thus clear that {\sc Sex-Equal Stable Marriage} seeks a stable matching that is fair towards both sides by minimizing the difference between their individual amounts of satisfaction. Unlike the {\sc Egalitarian Stable Marriage}, the {\sc Sex-Equal Stable Marriage} problem is known to be \NPH\ \cite{Kato93}. On the other hand, in {\sc Balanced Stable Marriage}, the objective is to find a stable matching that minimizes $\bal{\mu}=\max\{\sum_{(m,w)\in \mu}p_{m}(w), \sum_{(m,w)\in \mu} p_{w}(m)\}$ over $\mathrm{SM}$. At first sight, this measure might seem conceptually similar to the previous one, but in fact, the two measures are quite different. Indeed, {\sc Balanced Stable Marriage} does not attempt to find a stable matching that is fair, but one that is desirable by both sides. In other words, {\sc Balanced Stable Marriage} examines the amount of dissatisfaction of each party {\em individually}, and attempts to minimize the worse one among the two. This problem fits the common scenario in economics where each party is selfish in the sense that it desires a matching where its own dissatisfaction is minimized, irrespective of the dissatisfaction of the other party, and our goal is to find a matching desirable by both parties by ensuring that each individual amount of dissatisfaction does not exceed some threshold. In some situations, the minimization of $\bal{\mu}$ may indirectly also minimize $\delta(\mu)$, but in other situations, this may not be the case. Indeed,  McDermid \cite{mcdermidCom} constructed a family of instances where there does {\em not} exist any matching that is both a sex-equal stable matching and a balanced stable matching (the construction is also available in the book \cite{DBLP:books/ws/Manlove13}).

The {\sc Balanced Stable Marriage} problem was introduced in the influential work of Feder \cite{Feder95} on stable matchings. Feder \cite{Feder95} proved that this problem is \NPH\ and that it admits a 2-approximation algorithm. Later, it was shown that this problem also admits a $(2-1/\ell)$-approximation algorithm where $\ell$ is the maximum size of a set of acceptable partners \cite{DBLP:books/ws/Manlove13}. O'Malley \cite{omalleyThesis} phrased the {\sc Balanced Stable Marriage} problem in terms of constraint programming. Recently, McDermid and Irving \cite{McDermidIrving14} expressed interest in the design of fast exact exponential-time algorithms for {\sc Balanced Stable Marriage}. In this paper, we study this problem in the realm of fast exact exponential-time algorithms as defined by the field of Parameterized Complexity (see Section \ref{sec:prelim}). Recall that $\mathrm{SM}$ is the set of all stable matchings. In this context, we would like to remark that {\sc Egalitarian Stable Roommates} problem is NP-complete \cite{Feder95}. Recently, Chen et al.\cite{chen2017hard} showed that it is \FPT\ parameterized by the egalitarian cost. McDermid and Irving \cite{McDermidIrving14} showed that {\sc Sex-Equal Stable Marriage} is \NPH\ even if it is only necessary to decide whether the target $\Delta=\min_{\mu\in \mathrm{SM}}|\delta(\mu)|$ is 0 or not \cite{McDermidIrving14}. In particular, this means that {\sc Sex-Equal Stable Marriage} is not only W[1]-hard with respect to $\Delta$, but it is even paraNP-hard with respect to this parameter.\footnote{If a parameterized problem cannot be solved in polynomial time even when the value of the parameter is a fixed constant (that is, independent of the input), then the problem is said to be paraNP-hard.} In the case of {\sc Balanced Stable Marriage}, however, fixed-parameter tractability with respect to the target $\Bal=\min_{\mu\in \mathrm{SM}}\bal{\mu}$ trivially follows from the fact that this value is lower bounded by $\max\{|M|,|W|\}$.\footnote{We assume that any stable matching $\mu$ is perfect. We justify this assumption later.}

\paragraph{Our Contribution.}
We introduce two ``above-guarantee parameterizations'' of {\sc Balanced Stable Marriage}. To this end, we consider the minimum value $\Om$ of the total dissatisfaction of men that can be realized by a stable matching, and the minimum value $\Ow$ of the total dissatisfaction of women that can be realized by a stable matching. Formally, $\Om=\sum_{(m,w)\in \Mopt{}} p_m(w)$, and $\Ow = \sum_{(m,w)\in \Wopt{}} p_w(m)$, where $\Mopt{}$ and $\Wopt{}$ are the man-optimal and woman-optimal stable matchings, respectively. An input integer $k$ would indicate that our objective is to decide whether $\Bal \leq k$.
 In our first parameterization, the parameter is $k-\min\{\Om,\Ow\}$, and in the second one, it is $k-\max\{\Om,\Ow\}$. In other words, we would like to answer the following questions (recall that $\Bal=\min_{\mu\in \mathrm{SM}}\bal{\mu}$).

\defparproblemMy{{\sc Above-Min Balanced Stable Marriage} (\minBSM)}{An instance $(M,W, {\cal L}_M, {\cal L}_W)$ of {\sc Balanced Stable Marriage}, and a non-negative integer $k$.}{$t = k - \min\{\Om,\Ow\}$}{Is $\Bal\leq k$?}

\defparproblemMy{{\sc Above-Max Balanced Stable Marriage} (\maxBSM)}{An instance $(M,W, {\cal L}_M, {\cal L}_W)$ of {\sc Balanced Stable Marriage}, and a non-negative integer $k$.}{$t = k - \max\{\Om,\Ow\}$}{Is $\Bal \leq k$?}

Before stating our results, let us explain the choice of these parameterizations. Here, note that the best satisfaction the party of men can hope for is $\Om$, and the best satisfaction the party of women can hope for is $\Ow$. 
First, consider the parameter $t = k - \min\{\Om,\Ow\}$. Whenever we have a solution such that the amounts of satisfaction of {\em both} parties are {\em close enough} to the best they can hope for, this parameter is small. Indeed, the closer the satisfaction of both parties to the best they can hope for (which is exactly the case where both parties would find the solution desirable), the smaller the parameter is, and the smaller the parameter is, the faster a parameterized algorithm is. In other words, if there exists a solution that is desirable by both parties, our parameter is small.

In the parameterization above, as we take the {\em min} of $\{\Om,\Ow\}$, we need the satisfaction of {\em both} parties to be close to optimal in order to have a small parameter. As we are able to show that {\sc Balanced Stable Marriage} is \FPT\ with respect to this parameter, it is very natural to next examine the case where we take the {\em max} of $\{\Om,\Ow\}$. In this case, the closer the satisfaction of {\em at least one} party to the best it can hope for, the smaller the parameter is. In other words, now the demand from a solution---while not changing the definition of a solution in any way, that is, a solution is still a stable matching $\mu$ minimizing $\bal{\mu}$---in order to have a small parameter is {\em weaker}. In the jargon of Parameterized Complexity, it is said that the parameterization by $t = k - \max\{\Om,\Ow\}$ is ``above a higher guarantee'' than the parameterization by $t = k - \min\{\Om,\Ow\}$, since it is {\em always} the case that $\max\{\Om,\Ow\}\geq \min\{\Om,\Ow\}$.
Unfortunately, as we show in this paper, the parameterization by $k - \max\{\Om,\Ow\}$ results in a problem that is W[1]-hard. Hence, the complexities of the two parameterizations behave very differently. We remark that in Parameterized Complexity, it is {\em not at all} the rule that when one takes an ``above a higher guarantee'' parameterization, the problem would suddenly become W[1]-hard, as can be evidenced by the most classical above guarantee parameterizations in this field, which are of the {\sc Vertex Cover} problem. For that problem, three above guarantee parameterizations were considered in \cite{DBLP:conf/soda/GargP16,DBLP:journals/talg/LokshtanovNRRS14,DBLP:journals/toct/CyganPPW13,DBLP:conf/esa/RamanRS11}, each above a higher guarantee than the previous one that was studied, and each led to a problem that is \FPT. In that context, unlike our case, it is still not clear whether the bar can be raised higher. Overall, our results accurately draw the line between tractability and intractability with respect to the target value in the context of two very natural, useful parameterizations.

Finally, to be more precise, we note that our work proves three main theorems:
\begin{itemize}
\setlength{\itemsep}{-2pt}
\item First, we prove (in Section \ref{sec:kernel}) that \minBSM\ admits a kernel where {\em the number of people is linear in $t$}. For this purpose, we introduce notions that might be of independent interest in the context of a ``functional'' variant of \minBSM. Our kernelization algorithm consists of several phases, each simplifying a different aspect of \minBSM, and shedding light on structural properties of the \Yes-instances of this problem. Note that this result already implies that \minBSM\  is \FPT.
\item Second, we prove (in Section \ref{sec:alg}) that \minBSM\ admits a parameterized algorithm whose running time is {\em single exponential in the parameter $t$}. This algorithm first builds upon our kernel, and then incorporates the method of bounded search trees. 

\item  Third, we prove (in Section \ref{sec:hardness}) that \maxBSM\ is \WOH. This reduction is quite technical, and its importance lies in the fact that it rules out (under plausible complexity-theoretic assumptions) the existence of a parameterized algorithm for \maxBSM. Thus, we show that although \maxBSM\ seems quite similar to \minBSM, in the realm of Parameterized Complexity, these two problems are completely different. 
\end{itemize}


\section{Preliminaries}\label{sec:prelim}

Let $f$ be a function $f : A \rightarrow B$. For a subset $A'\subseteq A$, we let
$f|_{A'}$ denote the restriction of $f$ to $A'$. That is, $f|_{A'}:A'\rightarrow B$, and $f|_{A'}(a)=f(a)$ for all $a\in A'$.

Throughout the paper, whenever the instance $\I$ of {\sc Balanced Stable Marriage} under discussion is not clear from context or we would like to put emphasis on $\I$, we add ``$(\I)$'' to the appropriate notation. For example, we use the 
notation $t(\I)$ 
rather than $t$. When we would like to refer to the balance of a stable matching $\mu$ in a specific instance $\I$, we would use the notation $\mbal{}_{\I}(\mu)$. A matching is called a {\em perfect matching} if it matches every person (to some other person).

While designing our kernelization algorithm (see ``Parameterized Complexity''), we might be able to determine whether the input instance is a \Yes-instance or a \No-instance. For the sake of clarity, in the first case, we simply return \Yes, and in second case, we simply return \No. To properly comply with the definition of a kernel, the return of \Yes\ and \No\ should be interpreted as the return of a trivial \Yes-instance and a trivial \No-instance, respectively. Here, a trivial \Yes-instance can be the one in which $M=W=\emptyset$ and $k=0$, where the only stable matching is the one that is empty and whose balance is 0, and a trivial \No-instance can be the one where $M=\{m\}$, $W=\{w\}$, ${\cal A}(m)=\{w\}$, ${\cal A}(w)=\{m\}$ and $k=0$. 


\paragraph{Parameterized Complexity.}
A {\em parameterization} of a problem is the association of an integer $k$ with each input instance, which results in a {\em parameterized problem}. For our purposes, we need to recall three central notions that define the parameterized complexity of a parameterized problem. The first one is the notion of a {\em kernel}. Here, a parameterized problem is said to admit a {\em kernel} of size $f(k)$ for some function $f$ that depends {\em only} on $k$ if there exists a polynomial-time algorithm, called a {\em kernelization algorithm}, that translates any input instance into an equivalent instance of the same problem whose size is bounded by $f(k)$ and such that the value of the parameter does not increase. In case the function $f$ is polynomial in $k$, the problem is said to admit a {\em polynomial kernel}. Hence, kernelization is a mathematical concept that aims to analyze the power of preprocessing procedures in a formal, rigorous manner.
The second notion that we use is the one of {\em fixed-parameter tractability (\FPT)}. Here, a parameterized problem $\Pi$ is said to be \FPT\ if there is an algorithm that solves it in time $f(k)\cdot |I|^{\OO(1)}$, where $|I|$ is the size of the input and $f$ is a function that depends only on $k$. Such an algorithm is called a {\em parameterized algorithm}. In other words, the notion of \FPT\ signifies that it is not necessary for the combinatorial explosion in the running time of an algorithm for $\Pi$ to depend on the input size, but it can be confined to the parameter $k$.
Finally, we recall that Parameterized Complexity also provides tools to refute the existence of polynomial kernels and parameterized algorithms for certain problems (under plausible complexity-theoretic assumptions), in which context the notion of \WOH\ is a central one. It is widely believed that a problem that is \WOH\ is unlikely to be \FPT, and we refer the reader to the books \cite{ParamAlgorithms15b,ParamBook13} for more information on this notion in particular, and on Parameterized Complexity in general. The  notation $\OO^*$ is used to hide factors polynomial in the input size.

\noindent 
{\bf Reduction Rule.} To design our kernelization algorithm, we rely on the notion of a {\em reduction rule}. A reduction rule is a polynomial-time procedure that replaces an instance $(\I, k)$ of a parameterized problem 
${\rm \Pi}$  by a new instance $(\C{I'},k' )$ of ${\rm \Pi}$. Roughly speaking, a reduction rule is useful when the instance $\C{I'}$ is in some sense ``simpler'' than the instance $\I$. In particular, it is desirable to ensure that $k'\leq k$. The rule is said to be {\em safe} if $(\I,k)$ is a \Yes-instance if and only if  $(\C{I'},k' )$ is a \Yes-instance.

\paragraph{A Functional Variant of {\sc Stable Marriage}.}
To obtain our kernelization algorithm for \minBSM, it will be convenient  to work with a ``functional'' definition of preferences, resulting in a ``functional'' variant of this problem which we call \FBSM. Here, instead of the sets of preferences lists ${\cal L}_M$ and ${\cal L}_W$, the input consists of sets of preference {\em functions} ${\cal F}_M$ and ${\cal F}_W$, where ${\cal F}_M$ replaces ${\cal L}_M$ and ${\cal F}_W$ replaces ${\cal L}_W$. Specifically, every person $a \in M \cup W$ has an injective (one-to-one) function $f_a: \A(a) \To \Nb$, called a {\em preference function}. Intuitively, a lower function value corresponds to a higher preference. Since every preference function is injective, it defines a total order over a set of acceptable partners. Note that all of the definitions presented in the introduction extend to our functional variant in the natural way. For the sake of formality, we specify the required adaptations below. 

 
The input of the {\sc Functional Stable Marriage} problem consists of a set of men, $M$, and a set of women, $W$. Each person $a$ has a set of {\em acceptable partners}, denoted by ${\cal A}(a)$, and an injective function $f_a:{\cal A}(a)\rightarrow\N$ called a {\em preference function}. Without loss of generality, it is assumed that if a person $a$ belongs to the set of acceptable partners of a person $b$, then the person $b$ belongs to the set of acceptable partners of the person $a$. The set of preference functions of the men and the women are denoted by ${\cal F}_M$ and ${\cal F}_W$, respectively. A pair of a man and a woman, $(m,a)$, is an {\em acceptable pair} if both $m\in{\cal A}(w)$ and $w\in{\cal A}(m)$. Accordingly, the notion of a {\em matching} refers to a matching between men and women, where two people that are matched to one another form an acceptable pair. A matching $\mu$ {\em stable} if it does not have a {\em blocking pair}, which is an acceptable pair $(m,w)$ such that {\bf (i)} either $m$ is unmatched by $\mu$ or $f_{m}(w)< f_{m}(\mu(m))$, and {\bf (ii)} either $w$ is unmatched by $\mu$ or  $f_{w}(m)< f_{w}(\mu(w))$. The goal of the {\sc Functional Stable Marriage} problem is to find a stable matching.

The man-optimal stable matching, denoted by \Mopt, is a stable matching such that every stable matching $\mu$ satisfies the following condition: every man $m$ is either unmatched by both \Mopt\ and $\mu$ or $f_{m}(\Mopt{}(m))\leq f_{m}(\mu(m))$. The woman-optimal stable matching, denoted by \Wopt, is defined analogously. Given a stable matching $\mu$, define \[\bal{\mu}=\max\{\sum_{(m,w)\in \mu}f_{m}(w), \sum_{(m,w)\in \mu} f_{w}(m)\}.\] Moreover, $\Bal=\min_{\mu\in \mathrm{SM}}\bal{\mu}$, where SM is the set of all stable matchings, $\Om=\sum_{(m,w)\in \Mopt} f_m(w)$, and $\Ow = \sum_{(m,w)\in \Wopt} f_w(m)$. Finally, \FBSM\ is defined as follows.

\defparproblemMy{{\sc Above-Min Functional Balanced Stable Marriage} (\FBSM)}{An instance $(M,W, {\cal F}_M, {\cal F}_W)$ of {\sc Functional Balanced Stable Marriage}, and a non-negative integer $k$.}{$t = k - \min\{\Om,\Ow\}$}{Is $\Bal\leq k$?}

\medskip

From the above discussions it is straightforward to turn an instance of \minBSM\ into an equivalent instance of \FBSM\ as stated in the following observation.

\begin{observation}
\label{lem:basic}
Let $\C{I}=(M,W,\C{L}_{M},\C{L}_{W},k)$ be an instance of \minBSM. For each $a\in M\cup W$, define $f_{a}: \C{A}(a)\To \Nb$ by setting $f_{a}(b)=p_{a}(b)$ for all $b\in \C{A}(a)$. Then, $\C{I}$ is a \Yes-instance of \minBSM\ if and only if $(M,W, \F_{M}=\Sb{f_{m}}_{m\in M}, \F_{W}=\Sb{f_{w}}_{w\in W},k)$ is a \Yes-instance of \FBSM.
\end{observation}

\paragraph{Known Results.}
Finally, we state several classical results, which were originally presented in the context of {\sc Stable Marriage}. By their original proofs, these results also hold in the context of {\sc Functional Stable Marriage}. To be more precise,  given an instance of {\sc Functional Stable Marriage}, we can construct an equivalent instance of {\sc Stable Marriage}, by ranking the acceptable partners in the order of their function values, where a smaller value implies a higher preference. The instances are equivalent in the sense that they give rise to the exact same set of stable matchings. Hence, all the structural results about stable matchings in the usual setting (modeled by strict preference lists) apply to the generalized setting, modeled by injective functions.

\begin{proposition}[\cite{GS62}]
\label{lem:existenceSM}
For any instance of {\sc Stable Marriage} (or {\sc Functional Stable Marriage}), there exist a man-optimal stable matching, $\Mopt$, and a woman-optimal stable matching, $\Wopt$, and both $\Mopt$ and $\Wopt$ can be computed in polynomial time. 
\end{proposition}

The following powerful proposition is known as the Rural-Hospital Theorem.

\begin{proposition}[\cite{GS85a}]\label{lem:rht}
Given an instance of {\sc Stable Marriage} (or \FBSM), the set of men and women that are matched is the same for all stable matchings.
\end{proposition}

We further need a proposition regarding the man-optimal and woman-optimal stable matchings that implies Proposition~\ref{lem:rht}~\cite{GS85a}.

\begin{proposition}[\cite{DBLP:books/daglib/0066875}]
\label{lem:manwomanOpt}
For any instance of {\sc Stable Marriage} (or {\sc Functional Stable Marriage}), every stable matching $\mu$ satisfies the following conditions: every woman $w$ is either unmatched by both \Mopt\ and $\mu$ or $p_{w}(\Mopt{}(w))\geq p_{w}(\mu(w))$, and every man $m$ is either unmatched by both \Wopt\ and $\mu$ or $p_{m}(\Wopt{}(m))\geq p_{m}(\mu(m))$.
\end{proposition}


\section{Kernel}\label{sec:kernel}

In this section, we design a kernelization algorithm for \minBSM. More precisely, we prove the following theorem.

\begin{theorem}\label{thm:kernelBSM}
\minBSM\ admits a kernel that has at most $3t$ men, at most $3t$ women, and such that each person has at most $2t+1$ acceptable partners.
\end{theorem}

%



%
%
%



\subsection{Functional Balanced Stable Marriage}\label{sec:functional}

To prove Theorem \ref{thm:kernelBSM}, we first prove the following result for the \FBSM\ problem. 

\begin{lemma}\label{lem:kernelFBSM}
\pGBSM\ admits a kernel with at most $2t$ men, at most $2t$ women, and such that the image of the preference function of each person is a subset of $\{1,2,\ldots,t+1\}$.
\end{lemma}




To obtain the desired kernelization algorithm, we execute the following plan.
\begin{enumerate}
\setlength{\itemsep}{-1pt}
\item {\bf Cleaning Prefixes and Suffixes.} Simplify the preference functions by ``cleaning'' suffixes and thereby also ``cleaning'' prefixes.
\item {\bf Perfect Matching.} Zoom into the set of people matched by every stable matching.
\item {\bf Overcoming Sadness.} Bound the number of ``\sad'' people. Roughly speaking, a ``\sad'' person $a$ is one whose best attainable partner, $b$, does not reciprocate by considering $a$ as the best attainable partner.
\item {\bf Marrying Happy People.} Remove ``\happy'' people from the instance.
\item {\bf Truncating High-Values.} Obtain ``compact''   preference functions by truncating ``high-values''.   
\item {\bf Shrinking Gaps.} Shrink some of the gaps created by previous steps. 
\end{enumerate} 

Each of the following subsections captures one of the steps above. In what follows, we let $\I$ denote our current instance of \pGBSM. Initially, this instance is the input instance, but as the execution of our algorithm progresses, the instance is modified. The reduction rules that we present are applied {\em exhaustively} in the order of their presentation. In other words, at each point of time, the first rule whose condition is true is the one that we apply next. In particular, the execution terminates once the value of $t$ drops below $0$, as implied by the following rule.

\begin{rr}\label{rr:End1}
If $k<\max\{\Om,\Ow\}$, then return \No.
\end{rr}

\addtocounter{lemma}{1}

\begin{lemma}\label{lem:End1Safe}
Reduction Rule \ref{rr:End1} is safe.
\end{lemma}

\begin{proof}
For every $\mu\in\mathrm{SM}$, it holds that $\bal{\mu}\geq\max\{\Om,\Ow\}$. Thus, if $k<\max\{\Om,\Ow\}$, then every $\mu\in\mathrm{SM}$ satisfies $\bal{\mu}>k$. In this case, we conclude that $\Bal>k$, and therefore $\I$ is a \No-instance.
\end{proof}

Note that if $k<0$, then it also holds that $t<0$, and that if $t<0$, then $k<\min\{\Om,\Ow\}$. We remark that by Proposition \ref{lem:existenceSM}, it would be clear that each of our reduction rules can indeed be implemented in polynomial time.

\paragraph{Cleaning Prefixes and Suffixes.}
We begin by modifying the images of the preference functions. We remark that it is {\em necessary} to perform this step first as otherwise the following steps would not be correct. To clean prefixes while ensuring {\em both} safeness and that the parameter $t$ does not increase, we would actually need to clean {\em suffixes} first. Formally, we define suffixes as follows.

\begin{definition}
Let $(m, w)$ denote an acceptable pair. If $m$ is matched by $\Wopt$ and $f_{m}(w) > f_{m}(\Wopt(m))$, then we say that $w$ {\em belongs to the suffix of} $m$. Similarly, if $w$ is matched by $\Mopt$ and $f_{w}(m) > f_{w}(\Mopt(w))$, then we say that $m$ belongs to the suffix of~$w$. 
\end{definition}

By Proposition \ref{lem:manwomanOpt}, we have the following observation.

\begin{observation}\label{obs:suffix}
Let $(m, w)$ denote an acceptable pair such that one of its members belongs to the suffix of the other member. Then, there is no $\mu\in\mathrm{SM}(\I)$ that matches $m$ with $w$.
\end{observation}

For every person $a$, let $\worst(a)$ be the person in ${\cal A}(a)$ to whom $f_a$ assigns its worst preference value. More precisely, $\worst(a)=\mathrm{argmax}_{b\in{\cal A}(a)}f_a(b)$. We will now clean~suffixes.

\begin{rr}\label{rr:cleansuffix}\label{rule:cleanSuffix}
If there exists a person $a$ such that $\worst(a)$ belongs to the suffix of $a$, then define the preference functions as follows.
\begin{itemize}
\item $f'_{a}= f_{a}\restrict{\C{A}(a)\sm \{\worst(a)\}}$.
\item $f'_{\worst(a)}= f_{\worst(a)}\restrict{\C{A}(\worst(a))\sm\{a\}}$.
\item For all $b\in M\cup W\sm \{a, \worst(a)\}$: $f'_{b}= f_{b}$. 
\end{itemize}
The new instance is $\J=(M, W, \{f'_{m'}\}_{m'\in M}, \{f'_{w'}\}_{w'\in W}, k)$.
\end{rr}

\begin{lemma}\label{lem:cleanSuffix}
Reduction Rule \ref{rule:cleanSuffix} is safe, and $t({\cal I})=t({\cal J})$. 
\end{lemma}

\begin{proof}
By the definition of the new preference functions, we have that for every $\mu\in\mathrm{SM}(\I)\cap\mathrm{SM}(\J)$, it holds that $\sum_{(m,w)\in \mu}f_{m}(w)=\sum_{(m,w)\in \mu}f'_{m}(w)$ and $\sum_{(m,w)\in \mu}f_{w}(m)=\sum_{(m,w)\in \mu}f'_{w}(m)$. In particular, this means that to conclude that $\Bal(\I)=\Bal(\J)$ (which implies safeness) as well as that $\Om(\I)=\Om(\J)$ and $\Ow(\I)=\Ow(\J)$ (which implies that $t(\I)=t(\J)$), it is sufficient to show that SM$(\I)=\mathrm{SM}(\J)$. For this purpose, first consider some $\mu\in\mathrm{SM}(\I)$. By Observation \ref{obs:suffix}, it holds that $(a,\worst(a))\notin\mu$. Hence, $\mu$ is a matching in $\J$. Moreover, if $\mu$ has a blocking pair in $\J$, then by the definition of the new preference functions, it is also a blocking pair in $\I$. Since $\mu$ is stable in $\I$, we have that $\mu\in\mathrm{SM}(\J)$.

In the second direction, consider some $\mu\in\mathrm{SM}(\J)$. Then, it is clear that $\mu$ is a matching in $\I$. Moreover, if $\mu$ has a blocking pair $(m,w)$ in $\I$ that is not $(a,\worst(a))$, then $(m,w)$ is an acceptable pair in $\J$, and therefore by the definition of the new preference functions, we have that $(m,w)$ is also a blocking pair in $\J$. Hence, since $\mu$ is stable in $\J$, the only pair that can block $\mu$ in $\I$ is $(a,\worst(a))$. Thus, to show that $\mu\in\mathrm{SM}(\I)$, it remains to prove that $(a,\worst(a))$ cannot block $\mu$ in $\I$. Suppose, by way of contradiction, that $(a,\worst(a))$ blocks $\mu$ in $\I$. In particular, this means that $f_{a}(\worst(a)) < f_{a}(\mu(a))$. However, this contradicts the definition of $\worst(a)$.
\end{proof}

By cleaning suffixes, we actually also accomplish the objective of cleaning prefixes, which are defined as follows.

\begin{definition}
Let $(m, w)$ denote an acceptable pair. If $m$ is matched by $\Mopt$ and $f_{m}(w) < f_{m}(\Mopt(m))$, then we say that $w$ {\em belongs to the prefix of} $m$. Similarly, if $w$ is matched by $\Wopt$ and $f_{w}(m) < f_{w}(\Wopt(w))$, then we say that $m$ belongs to the prefix of~$w$. 
\end{definition}

Let us now claim that we have indeed succeeded in cleaning prefixes.

\begin{lemma}\label{lem:prefixCleaned}
Let $\I$ be an instance of \FBSM\ on which Reduction Rules \ref{rr:End1} to \ref{rule:cleanSuffix} have been exhaustively applied. Then, there does not exist an acceptable pair $(m,w)$ such that one of it members belongs to the prefix of the other one.
\end{lemma}

\begin{proof}
Suppose, by way of contradiction, that there exists an acceptable pair $(m,w)$ such that one of its members belongs to the prefix of the other one. Without loss of generality, assume that $w$ belongs to the prefix of $m$. Then, $f_m(w)<f_m(\Mopt(m))$. Since $\Mopt$ is a stable matching, it cannot be blocked by $(m,w)$, which means that $w$ is matched by $\Mopt$ and $f_w(\Mopt(w))<f_w(m)$. Thus, we have that $m$ belongs to the suffix of $w$, which contradicts the assumption that Reduction Rule \ref{rule:cleanSuffix} was applied exhaustively.
\end{proof}

As a direct result of this lemma, we have the following corollary.

\begin{corollary}\label{cor:prefixCleaned}
Let $\I$ be an instance of \FBSM\ on which Reduction Rules \ref{rr:End1} to \ref{rule:cleanSuffix} have been exhaustively applied. 
Then, for every acceptable pair $(m,w)$ in $\I$ where $m$ and $w$ are matched (not necessarily to each other) by both $\Mopt$ and $\Wopt$, it holds that $f_m(\Mopt(m))\leq f_m(w)\leq f_m(\Wopt(m))$ and $f_w(\Wopt(w))\leq f_w(m)\leq f_w(\Mopt(w))$.
\end{corollary}

\paragraph{Perfect Matching.}
Having Corollary \ref{cor:prefixCleaned} at hand, we are able to provide a simple rule that allows us to assume that every solution matches all people. 

\begin{rr}
\label{rr:basic}
If there exists a person unmatched by $\Mopt$, then let  $M'$ and $W'$ denote the subsets of men and women, respectively, who are matched by $\Mopt$. For each $a\in M'\cup W'$, denote $\C{A'}(a)=\C{A}(a) \cap (M' \cup W')$, and define $f'_{v}=f_v\restrict{\C{A'}(v)}$. The new instance is $\C{J}=(M', W', \{f'_{m}\}_{m\in M'}, \{f'_{w}\}_{w\in W'}, k)$. 
\end{rr}

To prove the safeness of this rule, we first prove the following lemma.

\begin{lemma}\label{lem:isolateUnmatch}
Let $\I$ be an instance of \FBSM\ on which Reduction Rules \ref{rr:End1} to \ref{rule:cleanSuffix} have been exhaustively applied. 
Then, for every person $a$ not matched by $\Mopt$, it holds that $\C{A}(a)=\emptyset$.
\end{lemma}

\begin{proof}
Let $a$ be a person not matched by $\Mopt$. Then, by Proposition~\ref{lem:rht}, it holds that $a$ is not matched by any stable matching. Hence, we can assume w.l.o.g.~that $a$ is man $m$. First, note that $\C{A}(m)$ cannot contain any woman $w$ that is not matched by some stable matching, else $(m,w)$ would have formed a blocking pair for that stable matching. Second, we claim that $\C{A}(m)$ cannot contain a woman $w$ that is matched by some stable matching. Suppose, by way of contradiction, that this claim is false. Then, by Proposition~\ref{lem:rht}, it holds that $\C{A}(m)$ contains a woman $w$ that is matched by $\Mopt$. We have that $w$ prefers $\Mopt(w)$ over $m$, else $(m,w)$ would have formed a blocking pair for $\Mopt$, which is impossible as $\Mopt$ is a stable matching. However, this implies that $m$ belongs to the suffix of $w$, which contradicts the supposition that Reduction Rule \ref{rule:cleanSuffix} has been exhaustively applied. We thus conclude that $\C{A}(a)=\emptyset$.
\end{proof}

\begin{lemma}\label{L:perfect-matching}
Reduction Rule~\ref{rr:basic} is safe, and $t({\cal I})=t({\cal J})$.
\end{lemma}

\begin{proof}By the definition of the new preference functions, we have that for every $\mu\in\mathrm{SM}(\I)\cap\mathrm{SM}(\J)$, it holds that $\sum_{(m,w)\in \mu}f_{m}(w)=\sum_{(m,w)\in \mu}f'_{m}(w)$ and $\sum_{(m,w)\in \mu}f_{w}(m)=\sum_{(m,w)\in \mu}f'_{w}(m)$. To conclude that the lemma is correct, it is thus sufficient to argue that SM$(\I)=\mathrm{SM}(\J)$. For this purpose, first consider some $\mu\in\mathrm{SM}(\I)$. By Proposition~\ref{lem:rht}, we have that $\mu$ is also a matching in $\J$.
Moreover, if $\mu$ has a blocking pair in $\J$, then by the definition of the new preference functions, it is also a blocking pair in $\I$. Since $\mu$ is stable in $\I$, we have that $\mu\in\mathrm{SM}(\J)$.

In the second direction, consider some $\mu\in\mathrm{SM}(\J)$. Then, it is clear that $\mu$ is a matching in $\I$. By Lemma \ref{lem:isolateUnmatch}, if $\mu$ has a blocking pair $(m,w)$ in $\I$, then both $m\in M'$ and $w\in W'$. However, for such a blocking pair $(m,w)$, we have that $(m,w)$ is an acceptable pair in $\J$, and therefore by the definition of the new preference functions, we have that $(m,w)$ is also a blocking pair in $\J$. Hence, since $\mu$ is stable in $\J$, we conclude that $\mu$ is also stable in $\I$.
\end{proof}

By Proposition~\ref{lem:rht}, from now onwards, we have that for the given instance, any stable matching is a perfect matching. Due to this observation, we can denote $n=|M|=|W|$, and for any stable matching $\mu$, we have the following equalities. \[\sum_{(m,w)\in \mu}f_{m}(w)=\sum_{m\in M}f_{m}(\mu(m));~~~~~~~~~~ \sum_{(m,w)\in \mu}f_{w}(m)=\sum_{w\in W}f_{w}(\mu(w)). \tag{I}\]

\paragraph{Overcoming Sadness.}
As every stable matching is a perfect matching, every person is matched by every stable matching, including the man-optimal and woman-optimal stable matchings. Thus, it is well defined to classify the people who do not have the same partner in the man-optimal and woman-optimal stable matchings as ``\bad''. That is, 
\begin{definition}\label{D:Bad}
A person $a\in M\cup W$ is  {\em \bad} if $\Mopt(a)\ne \Wopt(a)$.
\end{definition}

We let $\Mbad$ and $\Wbad$ denote the sets of \bad\ men and \bad\ women, respectively.  People who are not \bad\ are termed {\em \good}. Accordingly, we let $\Mgood$ and $\Wgood$ denote the sets of \good\  men and \good\  women, respectively. Note that $\Mbad= \emptyset$ if and only if $\Wbad= \emptyset$. Moreover, note that by the definition of $\Mopt$ and $\Wopt$, for a \happy\ person $a$ it holds that $a$ and $\Mopt(a)=\Wopt(a)$ are matched to one another by every stable matching. Let us now bound the number of \sad\ people in a \Yes-instance. 

\begin{rr}\label{rr:manybadmen} 
If $|\Mbad|>2t$ or $|\Wbad|>2t$, then return \No.
\end{rr}

\begin{lemma}\label{L:Bounding-Bad}
Reduction Rule~\ref{rr:manybadmen} is safe.
\end{lemma}

\begin{proof}We only prove that if $|\Mbad|>2t$, then $\I$ is a \No-instance, as the proof of the other case is symmetric to this one. Let us assume that $|\Mbad| > 2t$. Suppose, by way of contradiction, that $\I$ is a \Yes-instance. Then, there exists a stable matching $\mu$ such that $\bal{\mu}\leq k$. Partition $\Mbad=M'_S \uplus M''_S$ as follows. Set $M'_S$ to be the set of all $m$ in $M_S$ such that $m$ is not the partner of $\mu(m)$ in $\Wopt$.
$$M'_S=\{m\in \Mbad \,|\, f_{\mu(m)}(m)>f_{\mu(m)}(\Wopt(\mu(m)))\}.$$ Accordingly, set $M''_S=M_S\setminus M'_S$. Since $|\Mbad| > 2t$, at least one among $|M'_S|$ and $|M''_S|$ is (strictly) larger than $t$. Let us first handle the case where $|M'_S|>t$. Then,
\begin{eqnarray*}
\sum_{w \in W} f_w(\mu(w)) &= & \sum_{\{w\,|\,\mu(w) \in M'_S\}} f_w(\mu(w)) + \sum_{\{w\,|\,\mu(w) \notin M'_S\}} f_w(\mu(w))\\
   & \geq & \sum_{\{w\,|\,\mu(w) \in M'_S\}} \Tb{f_w(\Wopt{}(w)) +1} + \sum_{\{w\,|\,\mu(w) \notin M'_S\}} f_w(\Wopt{}(w))\\
   & = & \sum_{w\in W} f_w(\Wopt{}(w))  + \sum_{\{w\,|\,\mu(w) \in M'_S\}} 1 = \Ow +  |M'_S|\\
   & > & \Ow +t \geq k. 
 \end{eqnarray*}

Here, the first inequality followed directly from the definition of $M'_S$. As we have reached a contradiction, it must hold that $|M''_S|>t$. However, we now have that
\begin{eqnarray*}
\sum_{m \in M} f_m(\mu(m)) &=& \sum_{m\in M''_S} f_m(\mu(m))  +  \sum_{m\notin M''_S} f_m(\mu(m))\\
&\geq & \sum_{m\in M''_S} \Tb{f_{m}(\Mopt{}(m))+1} + \sum_{m\notin M''_S} f_m(\Mopt{}(m))\\
   & = & \sum_{m\in M} f_m(\Mopt{}(m))  + \sum_{\{m\,|\,\mu(m) \in M''_S\}} 1 = \Om +  |M''_S|\\
   & > & \Om +t \geq k. 
\end{eqnarray*}

Here, the first inequality followed from the definition of $M''_S$. Indeed, for all $m\in M''_S$, we have that $f_{\mu(m)}(m)\leq f_{\mu(m)}(\Wopt(\mu(m)))$, else $m$ would have belonged to $M'_S$. However, in this case we deduce that $m=\Wopt(\mu(m))$, and since $m\in M_S$, we have that $\mu(m)\neq \Mopt(m)$, which implies that $f_m(\mu(m))\geq f_{m}(\Mopt{}(m))+1$. As we have again reached a contradiction, we conclude the proof.
\end{proof}

\paragraph{Marrying Happy People.}
Towards the removal of \happy\ people, we first need to handle the special case where there are no \sad\ people. In this case, there is exactly one stable matching, which is the 
man-optimal stable matching (that is equal, in this case, to the woman-optimal stable matching). This immediately implies the safeness of the following rule. 

\begin{rr}\label{rr:onlygooodmen} 
If $\Mbad= \Wbad=\emptyset$, then return \Yes\ if $\bal{\Mopt} \leq k$ and \No\ otherwise.
\end{rr}

\begin{observation}\label{obs:onlygooodmen}
Reduction Rule~\ref{rr:onlygooodmen}  is safe.
\end{observation}

We now turn to discard \good\ people. When we perform this operation, we need to ensure that the balance of the instance is preserved. More precisely, we need to ensure that $\Bal(\I)=\Bal(\J)$, where $\J$ denotes the new instance resulting from the removal of some \good\ people. Towards this, we let $(m_{h},w_{h})$ denote a {\em \good\ pair}, which is simply a pair of a \good\  man and \good\  woman who are matched to each other in every stable matching. Then, we redefine the preference functions in a manner that allows us to transfer the ``contributions'' of $m_{h}$ and $w_{h}$ from 
$\Bal(\I)$ to $\Bal(\J)$ via some \bad\ man and woman. We remark that these \bad\ people exist because Reduction Rule \ref{rr:onlygooodmen} does not apply. The details are as follows. 

\begin{rr}\label{rr1}
If there exists a \good\ pair $(m_{h}, w_{h})$, then proceed as follows. Select an arbitrary \sad\ man $m_{s}$ and an arbitrary \sad\ woman $w_{s}$. Denote $M'=M\sm \{m_{h}\}$ and $W'=W\sm \{w_{h}\}$. For each person $a\in M'\cup W'$, the new preference function $f'_{a}:\A(a)\sm \{m_h,w_{h}\}\To \Nb$ is defined as follows.
\begin{itemize}
\item {\bf The preference function of $m_s$.} For each $w\in \C{A}(m_{s})\sm\{w_{h}\}$: $f'_{m_{s}}(w)= f_{m_{s}}(w)+f_{m_{h}}(w_{h})$. 
\item {\bf The preference function of $w_s$.} For each $m\in \C{A}(w_{s})\sm \{m_{h}\}$: $f'_{w_{s}}(m)= f_{w_{s}}(m)+f_{w_{h}}(m_{h})$. 
\item For each $w\in W'\sm \{w_{s}\}$: $f'_{w}= f_{w}\restrict{M'}$. 
\item For each  $m\in M'\sm \{m_{s}\}$: $f'_{m}=f_{m}\restrict{W'}$. 
\end{itemize}
The new instance is $\J=(M', W', \{f'_{m}\}_{m\in M'}, \{f'_{w}\}_{w\in W'}, k)$. 
\end{rr}

Let us first state a lemma concerning the proof of the forward direction of the safeness of Reduction Rule \ref{rr1}.

\begin{lemma}\label{lem:rr1Dir1}
Let $\mu\in\mathrm{SM}(\I)$. Then, $\mu'=\mu\setminus\{(m_h,w_h)\}$ is a stable matching in $\J$ such that $\mbal{}_{\J}(\mu')=\mbal{}_{\I}(\mu)$.
\end{lemma}

\begin{proof}
We first show that $\mu'\in\mathrm{SM}(\J)$. By Reduction Rule \ref{rr:basic}, it holds that $\mu$ is a perfect matching in $\I$. Since $(m_h,w_h)$ is a \happy\ pair, it is clear that $(m_h,w_h)\in\mu$, and therefore $\mu'$ is a perfect matching in $\J$. Let $(m,w)\notin \mu'$ be some acceptable pair in $\J$. Since $\mu\in\mathrm{SM}(\I)$ and it is a perfect matching, it holds that $f_{m}(w)>f_{m}(\mu(m))$ or $f_{w}(m)>f_{w}(\mu(w))$. Let us consider these two possibilities separately.
\begin{itemize}
\item Suppose that $f_{m}(w)>f_{m}(\mu(m))$. If $m\neq m_s$, then $f'_{m}(w)=f_{m}(w)$ and $f'_{m}(\mu'(m))=f_{m}(\mu(m))$, and therefore $f'_{m}(w)>f'_{m}(\mu'(m))$. Else, $f'_{m}(w)=f_{m}(w)+f_{m_{h}}(w_{h})$ and $f'_{m}(\mu'(m))=f_{m}(\mu(m))+f_{m_{h}}(w_{h})$, and therefore again $f'_{m}(w)>f'_{m}(\mu'(m))$.

\item Suppose that $f_{w}(m)>f_{w}(\mu(w))$. Analogously to the previous case, we get that $f'_{w}(m)>f'_{w}(\mu'(w))$.
\end{itemize}

Since the choice of $(m,w)$ was arbitrary, we conclude that $\mu'$ does not have a blocking pair in $\J$, and therefore $\mu'\in\mathrm{SM}(\J)$. To show that $\mbal_{\J}(\mu')=\mbal_{\I}(\mu)$, note that

\begin{eqnarray*}
\mbal_{\J}(\mu') &=& \max\{\sum_{m\in M\setminus \{m_h\}}f'_m(\mu'(m)), \sum_{w\in W\setminus \{w_h\}}f'_w(\mu'(w))\}\\
   &=& \max\{f'_{m_s}(\mu'(m_s))+\sum_{m\in M\setminus \{m_h,m_s\}}f'_m(\mu'(m)),\\
	&& \hspace{8.5em}
	f'_{w_s}(\mu'(w_s))+\sum_{w\in W\setminus \{w_h,w_s\}}f'_w(\mu'(w))\}\\
   &=& \max\{f_{m_s}(\mu(m_s))+f_{m_{h}}(w_{h})+\sum_{m\in M\setminus \{m_h,m_s\}}f_m(\mu(m)),\\
	&& \hspace{8.5em}
	f_{w_s}(\mu(w_s))+f_{w_{h}}(m_{h})+\sum_{w\in W\setminus \{w_h,w_s\}}f_w(\mu(w))\}\\
   &=& \max\{\sum_{m\in M}f_m(\mu(m)), \sum_{w\in W}f_w(\mu(w))\} = \mbal_{\I}(\mu).
\end{eqnarray*}

This concludes the proof.
\end{proof}

To prove a lemma addressing the reverse direction, we first state the following observation, whose correctness follows directly from Corollary \ref{cor:prefixCleaned}.

\begin{observation}\label{lemma:happyPairIsolated}
Let $\I$ be an instance of \FBSM\ on which Reduction Rules \ref{rr:End1} to \ref{rule:cleanSuffix} have been exhaustively applied. 
Then, for every happy pair $(m_h,w_h)$, it holds that $\C{A}(m_h)=\{w_h\}$ and $\C{A}(w_h)=\{m_h\}$.
\end{observation}

\begin{lemma}\label{lem:rr1Dir2}
Let $\mu'\in\mathrm{SM}(\J)$. Then, $\mu=\mu'\cup\{(m_h,w_h)\}$ is a stable matching in $\I$ such that $\mbal{}_{\I}(\mu)=\mbal{}_{\J}(\mu')$.
\end{lemma}

\begin{proof}We first show that $\mu\in\mathrm{SM}(\I)$. By Reduction Rule \ref{rr:basic}, it holds that $\mu'$ is a perfect matching in $\J$. Since $(m_h,w_h)$ is a \happy\ pair and by Observation \ref{lemma:happyPairIsolated}, we have that $\mu$ is a perfect matching in $\I$ such that neither $m_h$ nor $w_h$ participate in any pair that blocks $\mu$ (if such a pair exists). Let $(m,w)\notin\mu'$ be some acceptable pair in $\I$ such that $m\neq m_h$ and $w\neq w_h$. Since $\mu'\in\mathrm{SM}(\J)$ and it is a perfect matching, it holds that $f'_{m}(w)>f'_{m}(\mu'(m))$ or $f'_{w}(m)>f'_{w}(\mu'(w))$. Let us consider these two possibilities separately.
\begin{itemize}
\item Suppose that $f'_{m}(w)>f'_{m}(\mu(m))$. If $m\neq m_s$, then $f_{m}(w)=f'_{m}(w)$ and $f_{m}(\mu(m))=f'_{m}(\mu'(m))$, and therefore $f_{m}(w)>f_{m}(\mu(m))$. Else, $f_{m}(w)=f'_{m}(w)-f_{m_{h}}(w_{h})$ and $f_{m}(\mu(m))=f'_{m}(\mu'(m))-f_{m_{h}}(w_{h})$, and therefore again $f_{m}(w)>f_{m}(\mu(m))$.

\item Suppose that $f'_{w}(m)>f'_{w}(\mu'(w))$. Analogously to the previous case, we get that $f_{w}(m)>f_{w}(\mu(w))$.
\end{itemize}

Since the choice of $(m,w)$ was arbitrary, we conclude that $\mu$ does not have a blocking pair in $\I$, and therefore $\mu\in\mathrm{SM}(\I)$. To show that $\mbal_{\I}(\mu)=\mbal_{\J}(\mu')$, we follow the exact same argument as the one present in the proof of Lemma~\ref{lem:rr1Dir1}.
\end{proof}

We now turn to justify the use of Reduction Rule \ref{rr1}.

\begin{lemma}\label{L:RR1-safe}
Reduction Rule \ref{rr1} is safe, and $t({\cal I})=t({\cal J})$.
\end{lemma}

\begin{proof}By Lemmata \ref{lem:rr1Dir1} and \ref{lem:rr1Dir2}, it holds that $\Bal(\I)=\Bal(\J)$, and that both $\mbal_{\I}(\Mopt(\I))=\mbal{}_{\J}(\Mopt(\J))$ and $\mbal_{\I}(\Wopt(\I))=\mbal{}_{\J}(\Wopt(\J))$. As the argument $k$ remained untouched, we have that Reduction Rule \ref{rr1} is safe as well as that $t({\cal I})=t({\cal J})$.
\end{proof}

Before we proceed to examine preference functions more closely, let us take a step back and prove the following result.

\begin{lemma}\label{Cor:bad-men-women}
Given an instance $\I$\ of \pGBSM, one can exhaustively apply Reduction Rules~\ref{rr:End1} to \ref{rr1} in polynomial time to obtain an instance $\J$ such that $t(\J)\leq t(\I)$. All people in $\J$ are \sad\ and matched by every stable matching, and there exist at most $2t$ men and at most $2t$ women. 
\end{lemma}

\begin{proof}Notice that each rule among Reduction Rules~\ref{rr:End1} to \ref{rr1} can be applied in polynomial time, and it either terminates the execution of the algorithm or shrinks the size of the instance. Hence, it is clear that the instance $\J$ is obtained in polynomial time, and the claim that $t(\J)\leq t(\I)$ follows from Lemmata \ref{L:perfect-matching} and \ref{L:RR1-safe}. Due to Reduction Rules \ref{rr:basic} and \ref{rr1}, we have that all people in $\J$ are \sad\ and matched by every stable matching. Thus, due to Reduction Rule \ref{rr:manybadmen}, we also have that there exist at most $2t$ men and at most $2t$ women.
\end{proof}


\paragraph{Truncating High-Values.}
Up until now, we have bounded the number of people. However, the images of the preference functions can contain integers that are not bounded by a function polynomial in the parameter. Thus, even though the number of people is upper bounded by $4t$, the total size of the instance can be huge. Accordingly, in what follows, we need to process the images of the preference functions. Recall that we have already modified preference functions so that they would not contain irrelevant information in their {\em prefixes} and {\em suffixes}.
Our current goal is to truncate ``high-values''' of preference functions. To understand the intuition behind the rule we present next, suppose that there exists a stable matching $\mu$ and a man $m$ such that $f_{m}(\mu(m)) > t + f_{m}(\Mopt(m))$. That is, $\mu$ matches $m$ to a woman whose position is larger than the position of the woman with whom $m$ is matched by \Mopt\ by at least $t$ units. Then, $\bal{\mu}\geq \sum_{m'\in M}f_{m'}(\mu(m')) > \Om + t \geq k$. Hence, irrespective of whether or not the current instance is a \Yes-instance, we know that $\mu$ is not a yes-certificate. We thus observe that we should delete all those acceptable pairs whose presence in any 
stable matching  prevents its balance from being upper bounded by $k$. Formally,

\begin{rr}\label{rr:Truncation}\label{rr2}
If there exists an acceptable pair $(m,w)$ such that $f_{m}(w)>(k-\Om)+f_{m}(\Mopt(m))$ or $f_{w}(m)>(k-\Ow)+f_{w}(\Wopt(w))$, then define the preference functions as follows: 
\begin{itemize}
\item $f'_{m}=f_{m}\restrict{\A(m)\sm \{w\}}$. 
\item $f'_{w} = f_{w}\restrict{\A(w)\sm \{m\}}$. 
\item For all $a\in M\cup W \sm \{m, w\}$: $f'_{a}=f_{a}$. 
\end{itemize}
The new instance is $\J=(M, W, \{f'_{m'}\}_{m'\in M}, \{f'_{w'}\}_{w'\in W}, k)$. 
\end{rr}

\begin{lemma}\label{lem:trunc}
Reduction Rule \ref{rr:Truncation} is safe, and $t({\cal I})\geq t({\cal J})$.
\end{lemma}

\begin{proof}
Without loss of generality, suppose that $f_{m}(w)>(k-\Om(\I))+f_{m}(\Mopt(m))$. Due to Reduction Rule \ref{rr:End1}, we have that $k-\Om(\I)\geq 0$, and therefore $f_m(w)>f_{m}(\Mopt(m))$, which implies that $w\neq \Mopt(m)$. We thus have that $\Mopt(\I)$ is also a matching in $\J$, and due to Corollary \ref{cor:prefixCleaned}, we deduce that $\Mopt(\I)=\Mopt(\J)$. 
First, we would like to show that $t(\I)\geq t(\J)$. For this purpose, it is sufficient to show that $\Om(\I)\leq \Om(\J)$ and $\Ow(\I)\leq \Ow(\J)$. Since $\Mopt(\I)=\Mopt(\J)$, it is clear that $\Om(\I)=\Om(\J)$. By Reduction Rule \ref{rr:basic}, $\Mopt(\I)$ matches all people in $\I$.
 Thus, by Proposition \ref{lem:rht} and since $\Mopt(\I)=\Mopt(\J)$, we have that $\Wopt(\J)$ matches all people in $\J$, and hence all people in $\I$. By Corollary \ref{cor:prefixCleaned}, for any woman $w$ and the man $m$ most preferred by $w$ in $\J$, it holds that $w$ does not prefer $m$ over $\Wopt(w)$ in $\I$. Thus, by the definition of the preference functions, we have that $\Ow(\I)\leq \Ow(\J)$.

To show that the rule is safe, we need to show that $\Bal(\I)\leq k$ if and only if $\Bal(\J)\leq k$. For this purpose, let us first suppose that $\Bal(\I)\leq k$. Then, there exists $\mu\in\mathrm{SM}(\I)$ such that $\mbal_{\I}(\mu)\leq k$. Notice that if $\mu(m) = w$, then since $f_{m}(w)>(k-\Om)+f_{m}(\Mopt(m))$ and by the equations in ``Perfect Matching'', we have that
\[\begin{array}{ll}
k&\geq {\mbal}_{\I}(\mu)\\
&\geq \displaystyle{\sum_{m'\in M}f_{m'}(\mu(m'))}\\
&\geq (k-\Om(\I)) + \displaystyle{\sum_{m'\in M}f_{m'}(\Mopt(m'))} > (k-\Om(\I)) + \Om(\I),
\end{array}\]
which is a contradiction. Hence, $(m,w)\notin \mu$. Therefore, $\mu$ is a matching in $\J$. By the definition of the new preference functions, if $\mu$ has a blocking pair in $\J$, then this pair also blocks $\mu$ in $\I$. Since $\mu$ is stable in $\I$, we have that $\mu$ is also stable in $\J$. Now, by our definition of the new preference functions, we have that $\mbal_{\I}(\mu)=\mbal_{\J}(\mu)$. Since $\mbal_{\I}(\mu)\leq k$, we thus conclude that $\Bal(\J)\leq k$.

In the second direction, suppose that $\Bal(\J)\leq k$. Then, there exists $\mu\in\mathrm{SM}(\J)$ such that $\mbal_{\J}(\mu)\leq k$. Clearly, $\mu$ is also a matching in $\I$.  Moreover, by the definition of the preference functions, every acceptable pair in $\I$ that is also present in $\J$ cannot block $\mu$ in $\I$, else it would have also blocked $\mu$ in $\J$. Thus, if $\mu$ has a blocking pair in $\I$, then this pair must be $(m,w)$. We claim that $(m,w)$ cannot block $\mu$ in $\I$, which would imply that $\mu\in\mathrm{SM}(\I)$. Suppose, by way of contradiction, that this claim is not true. Recall that we have already proved that $\Mopt(\I)=\Mopt(\J)$. Let us denote $\Mopt=\Mopt(\I)$.
We have that $\mu$ matches $m$, which implies that $f_m(\mu(m))>f_m(w)$. Since $f_{m}(w)>(k-\Om(\I))+f_{m}(\Mopt(m))$, we deduce that $f_m(\mu(m))>(k-\Om(\I))+f_{m}(\Mopt(m))$. Furthermore, since $f'_m(\mu(m))=f_m(\mu(m))$, $f'_{m}(\Mopt(m))=f_{m}(\Mopt(m))$ and $\Om(\I)=\Om(\J)$, we get that $f'_m(\mu(m))>(k-\Om(\J))+f'_{m}(\Mopt(m))$. However, we then have that
\[\begin{array}{ll}
k&\geq {\mbal}_{\J}(\mu)\\
&\geq \displaystyle{\sum_{m'\in M}f'_{m'}(\mu(m'))}\\
& \geq (k-\Om(\J)) + \displaystyle{\sum_{m'\in M}f'_{m'}(\Mopt(m'))} > (k-\Om(\J)) + \Om(\J),
\end{array}\]
which is a contradiction. Therefore, $\mu\in\mathrm{SM}(\I)$. The definition of the new preference functions imply that $\mbal_{\I}(\mu)=\mbal_{\J}(\mu)$. Since $\mbal_{\J}(\mu)\leq k$, we conclude that $\Bal(\I)\leq k$.
\end{proof}

\paragraph{Shrinking Gaps.}
Currently, there might still exist a man $m$ or a woman $w$ such that $f_{m}(\Mopt(m))>1$ or $f_{w}(\Wopt(w))>1$, respectively.
In the following rule, we would like to decrease some values assigned by the preference functions of such men and women in a manner that preserves equivalence.

\begin{rr}\label{rr:PrefixShrinking}
If there exist $m\in M$ and $w\in W$ such that $f_{m}(\Mopt(m)) >1$ and $f_{w}(\Wopt(w))>1$, then define the preference functions as follows.
\begin{itemize} 
\item {\bf The preference function of $m$.} For all $w'\in \A(m)$: $f'_{m}(w')=f_{m}(w')-1$. 
\item {\bf The preference function of $w$.} For all $m'\in \A(w)$: $f'_{w}(m')=f_{w}(m')-1$.
\item For all $a\in M\cup W\sm \{m, w\}$: $f'_{a}=f_{a}$.
\end{itemize}
The new instance is $\J=(M,W,\{f'_{m'}\}_{m' \in M}, \{f'_{w'}\}_{w'\in W},k-1)$.
\end{rr}

\begin{lemma}\label{L:PrefixShrinking-safe}
Reduction Rule \ref{rr:PrefixShrinking} is safe, and $t({\cal I})=t({\cal J})$. 
\end{lemma}

\begin{proof}
Let us first observe that the set of acceptable pairs in $\I$ is equal to the set of acceptable pairs of $\J$. Furthermore, for every person $a$, including the cases where this person is either $m$ or $w$, any two acceptable partners $b$ and $b'$ of $a$ that satisfy $f_a(b)<f_a(b')$ also satisfy $f'_a(b)<f'_a(b')$, and vice versa. Indeed, this observation follows directly from our definition of the new preference functions. In other words, if a person prefers some person over another in $\I$, then this person also has the same preference order in $\J$, and vice versa. We thus deduce that SM$(\I)=\mathrm{SM}(\J)$.

We proceed by claiming that for all $\mu\in\mathrm{SM}(\I)$, we have that $\mbal_{\I}(\mu)=\mbal_{\J}(\mu)+1$. Indeed, by the definition of the new preference functions and the equations in ``Perfect Matching'', we have that

\[\begin{array}{ll}
{\mbal}_{\I}(\mu)  &= \displaystyle{\max\{\sum_{m'\in M}f_{m'}(\mu(m')), \sum_{w'\in W} f_{w'}(\mu(w'))\}}\\

 &= \max\{f_{m}(\mu(m)) + \displaystyle{\sum_{m'\in M\setminus\{m\}}f_{m'}(\mu(m')), f_{w}(\mu(w)) + \sum_{w'\in W\setminus\{w\}} f_{w'}(\mu(w'))\}}\\

&= \max\{f'_{m}(\mu(m))+1 + \displaystyle{\sum_{m'\in M\setminus\{m\}}f'_{m'}(\mu(m')), f'_{w}(\mu(w))+1 + \sum_{w'\in W\setminus\{w\}} f'_{w'}(\mu(w'))\}}\\

&= \max\{\displaystyle{\sum_{m'\in M}f'_{m'}(\mu(m')),\sum_{w'\in W} f'_{w'}(\mu(w'))\}}+1 ={\mbal}_{\J}(\mu)+1.
\end{array}
\]

Furthermore, the arguments above also show that $\Om(\I)=\Om(\J)+1$ and $\Ow(\I)=\Ow(\J)+1$. Hence, we have that $\Bal(\I)=\Bal(\J)+1$. Since $k$ was decreased by $1$, we conclude that the rule is safe and that $t(\I)=t(\J)$.
\end{proof}

The usefulness of Reduction Rule~\ref{rr:PrefixShrinking} lies in the observation that after the exhaustive application of this rule, at least one of the two parties does not have any member without a person assigned 1 by his/her preference function. More precisely,

\begin{observation}\label{obs:funcGap}
Let $\I$\ be an instance of \pGBSM\ that is reduced  with respect to Reduction Rules~\ref{rr:End1} to \ref{rr:PrefixShrinking}. Then, either {\bf (i)} for every $m\in M$, we have that $f_{m}(\Mopt(m))=1$, or {\bf (ii)} for every $w\in W$, we have that $f_{w}(\Wopt(w))=1$. In particular, either {\bf (i)} $\Om=|M|$ or {\bf (ii)} $\Ow=|W|$.
\end{observation}






This concludes the description of our reduction rules. We are now ready to prove Lemma~\ref{lem:kernelFBSM}. 



\begin{proof}[Proof of Lemma~\ref{lem:kernelFBSM}] Given an instance $\I$ of \FBSM, our kernelization algorithm exhaustively applies Reduction Rules \ref{rr:End1} to \ref{rr:PrefixShrinking}, after which it outputs the resulting instance, $\J$, as the kernel. Notice that each rule among Reduction Rules~\ref{rr:End1} to \ref{rr:PrefixShrinking} can be applied in polynomial time, and it either terminates the execution of the algorithm or shrinks the size of the instance. Hence, it is clear that the instance $\J$ is obtained in polynomial time. The claims that $\J$ and $\I$ are equivalent and that $t(\J)\leq t(\I)$ follow directly from the lemmata that prove the safeness of each rule as well as argue with respect to the parameter. By Lemma \ref{Cor:bad-men-women}, we also have that the instance contains at most $2t$ (sad) men and at most $2t$ (sad) women.

It remains to show that the image of the preference function of each person is a subset of $\{1,2,\ldots,t+1\}$. By Reduction Rule \ref{rr2}, every acceptable pair $(m,w)$ satisfies $f_{m}(w)\leq (k-\Om)+f_{m}(\Mopt(m))$ and $f_{w}(m)\leq (k-\Ow)+f_{w}(\Wopt(w))$. Moreover, for every acceptable pair $(m,w)$, it holds that $f_{m}(\Mopt(m))\leq \Om-(|M|-1)$ and $f_{w}(\Wopt(w))\leq \Ow-(|W|-1)$. Thus, every acceptable pair $(m,w)$ satisfies $f_{m}(w)\leq k-(|M|-1)$ and $f_{w}(m)\leq k-(|W|-1)$. By Reduction Rule \ref{rr:basic} and Observation \ref{obs:funcGap}, we have that $t=k-\min\{\Om,\Ow\}=k-|M|=k-|W|$. Hence, we further conclude that every acceptable pair $(m,w)$ satisfies $f_{m}(w)\leq t+1$ and $f_{w}(m)\leq t+1$.
\end{proof}

\subsection{Balanced Stable Marriage}

Having proved Lemma \ref{lem:kernelFBSM}, we have a kernel for \pGBSM. We would like to employ this kernelization algorithm to design one for \pBSM. For this purpose, we need to remove gaps from preference functions. Once we do this, we can view preference functions as preference lists and obtain the desired kernel. In what follows, we describe our kernelization algorithm for \pBSM.

Let ${\cal K}=(M',W',{\cal L}_{M'},{\cal L}_{W'},k')$ be the input instance, which is an instance of \pBSM. Our algorithm begins by applying the reduction given by Observation \ref{lem:basic} to translate ${\cal K}$ into an instance ${\I}'=(M',W',{\cal F}_{M'},{\cal F}_{W'},k')$ of \pGBSM. Then, our algorithm applies the kernelization algorithm given by Lemma \ref{lem:kernelFBSM} to ${\I}'$, obtaining a reduced instance ${\I}=(M,W,{\cal F}_M,{\cal F}_W,k)$ of \pGBSM. By Lemma \ref{lem:kernelFBSM}, this instance has at most $2t$ men, at most $2t$ women, and the image of the preference function of each person is a subset of $\{1,2,\ldots,t+1\}$.  To eliminate ``gaps'' in the preference functions, the algorithm proceeds as described below. Note that we no longer apply any reduction rule from Section \ref{sec:functional} (even if its condition is satisfied), as we currently give a {\em new} kernelization procedure rather than an extension of the previous one. Let us first formally define the notion of a gap.

\begin{definition}\label{D:gaps}
Let $a\in M\cup W$, and $i$ be positive integer outside the image of $f$. If there exists an integer $j>i$ that belongs to the image of $f$, then $f_{a}$ is said to have a {\em gap} at $i$.
\end{definition}


\paragraph{Inserting Dummies.}
We have ensured that the largest number in the image of any preference function is at most $t+1$. As every person is \sad, it has at least two acceptable partners, and hence it has at most $t-1\leq t$ gaps. To handle the gaps of {\em all} people, we create a set of $t$ dummy men and $t$ dummy women. Our objective is to introduce these dummy people as acceptable partners for people who have gaps in their preference functions, such that the function values of the dummy people would fill the gaps. In the context of the following rule, note that currently there are no \happy\ people, and hence this rule would be applied (only~once).

\begin{rr}\label{rr:add-dummies}
If there do not exist \happy\ people, then let $X=\{x_1, x_2, \hdots, x_{t}\}$ denote a set of new (dummy) men, and  $Y=\{y_1, y_2, \hdots, y_{t}\}$ denote a set of new (dummy) women. For each $i \in \{1,2,\ldots,t\}$, initialize $\A(x_{i})=\{y_{i}\}$, $\A(y_{i})=\{x_{i}\}$ and $f_{x_i}(y_i) = f_{y_i}(x_i) = 1$. The new instance is $\J=(M\cup X, W\cup Y, \{f_{m}\}_{m\in M\cup X}, \{f_{w}\}_{w\in W\cup Y}, k+t)$.
\end{rr}

We note that for all $i\in\{1,2,\ldots,t\}$, it holds that $(x_{i}, y_{i})$ is a \happy\ pair.

\begin{lemma}\label{lem:add-dummies}
Reduction Rule \ref{rr:add-dummies} is safe, and $t(\I)=t(\J)$.
\end{lemma}

\begin{proof}
For all $i\in\{1,2,\ldots,t\}$, it holds that $(x_{i}, y_{i})$ is a \happy\ pair, and therefore it is present in every stable matching in $\J$. By our definition of the new preference functions, it is clear that if $\mu$ is a stable matching in $\I$, then $\mu'=\mu\cup\{(x_1,y_1),\ldots,(x_t,y_t)\}$ is a stable matching in $\J$. Moreover, if $\mu'$ is a stable matching in $\J$, then $\mu=\mu'\setminus\{(x_1,y_1),\ldots,(x_t,y_t)\}$ is a stable matching in $\I$. Hence, since for all $i\in\{1,2,\ldots,t\}$, it holds that $f_{x_i}(y_i) = f_{y_i}(x_i) = 1$, our definition of the new preference functions directly implies that $\Bal(\I)+t=\Bal(\J)$, $\Om(\I)+t=\Om(\J)$ and $\Ow(\I)+t=\Ow(\J)$. Hence, $t(\I)=t(\J)$, which concludes the proof.
\end{proof}

\begin{rr}\label{rr:Fill-Dummy}{\rm [Male version]}\label{rr6}
If there exists $m\in M$ such that $f_{m}$ has a gap at some $j$, then select some $y_{i}\in Y \sm \A(m)$, and set $\A'(m)=\A(m)\cup \{y_{i}\}$ and $\A'(y_{i})=\A(y_{i}) \cup \{m\}$. The preference functions are defined as follows.
\begin{itemize}
\item {\bf The preference function of $m$.} $f'_{m}(y_{i})=j$, and for all $a\in\A(m)$, $f'_{m}(a)=f_m(a)$.
\item {\bf The preference function of $y_i$.} $\displaystyle{f'_{y_{i}}(m)=\max_{m' \in \C{A}(y_i)}(f_{y_{i}}(m')+1})$, and for all $a\in\A(y_i)$, $f'_{y_i}(a)=f_{y_i}(a)$.
\item For all $a\in (M\cup W)\setminus\{m,y_i\}$: $f'_a=f_a$.
\end{itemize}
The new instance is $\J=(M, W, \{f'_{m'}\}_{m'\in M}, \{f'_{w'}\}_{w'\in W}, k)$.
\end{rr}

\begin{lemma}\label{lem:dummyfilling}
Reduction Rule \ref{rr:Fill-Dummy} is safe, and $t(\I)=t(\J)$.
\end{lemma}

\begin{proof}The only modifications that are performed are the insertion of $m$ intro the set of acceptable partners of $y_i$ as the least preferred person, and the insertion of $y_i$ into the set of acceptable partners of $m$ in a location that previosuly contained a gap. Let us first observe that since $f_{x_i}(y_i) = f_{y_i}(x_i) = 1$ and $f'_{x_i}(y_i) = f'_{y_i}(x_i) = 1$, it holds that $(x_i,y_i)$ is a \happy\ pair in both $\I$ and $\J$. Hence, it is clear that SM$(\I)=\mathrm{SM}(\J)$, $\Om(\I)=\Om(\J)$, $\Ow(\I)=\Ow(\J)$, and that the balance of any stable matching in $\I$ is equal to its balance in $\J$. We thus conclude that the rule is safe and that $t(\I)=t(\J)$.
\end{proof}

Analogously, we have a female version of Reduction Rule \ref{rr:Fill-Dummy} where we fill a gap in the preference function of some woman $w\in W$. We do not repeat our arguments again, and straightaway state the following result, which follows directly from the safeness of Reduction Rule \ref{rr:add-dummies} and the male and female versions of Reduction Rule \ref{rr6}. 

\begin{lemma}\label{Cor:Kernel-pGBSM-without-gaps}
\pGBSM\ admits a kernel that has at most $3t$ men among whom at most $2t$ are \sad, at most $3t$ women among whom at most $2t$ are \sad, and such that each \happy\ person has at most $2t+1$ acceptable partners and each \sad\ person has at most $t+1$ acceptable partners.
\end{lemma}

Finally, we translate the kernel for \pGBSM\ to an instance of \pBSM\ as follows. For all $a\in M\cup W$ and $b\in{\cal A}(a)$, we set $p_a(b)=f_a(b)$. The new instance is $\J=(M, W, \{p_{m}\}_{m\in M}, \{f_{w}\}_{w\in W}, k)$. Clearly, we thus obtain an equivalent instance, which leads us to the following somewhat stronger version of Theorem \ref{thm:kernelBSM}, relevant to Appendix~\ref{sec:alg}.

\begin{lemma}\label{lem:kernelBSM}
\minBSM\ admits a kernel that has at most $3t$ men among whom at $2t$ are \sad, at most $3t$ women among whom at most $2t$ are \sad, and such that each \happy\ person has at most $2t+1$ acceptable partners and each \sad\ person has at most $t+1$ acceptable partners. Moreover, every stable matching in the kernel is a perfect matching.
\end{lemma}

This concludes the proof of Theorem \ref{thm:kernelBSM}.

\section{Parameterized Algorithm}\label{sec:alg}

In this section, we design a parameterized algorithm for \minBSM. More precisely, we prove the following theorem. 

\setcounter{thm}{1}
\begin{theorem}\label{thm:alg}
\minBSM\ can be solved in time $\OO^*(8^t)$.
\end{theorem}

As our algorithm is based on the method of bounded search trees, we first give a brief description of this technique.

\subsection{Bounded Search Tree}\label{sec:bounded}

The running time of an algorithm that uses bounded search trees can be analyzed as follows (see, e.g.,~\cite{ParamAlgorithms15b}). Suppose that the algorithm executes a branching rule which has $\ell$ branching options (each leading to a recursive call with the corresponding parameter value), such that, in the $i^\mathrm{th}$ branch option, the current value of the parameter decreases by $b_i$. Then, $(b_1,b_2,\ldots,b_{\ell})$ is called the {\em branching vector} of this rule. We say that $\alpha$ is the {\em root} of $(b_1,b_2,\ldots,b_{\ell})$ if it is the (unique) positive real root of $x^{b^*} = x^{b^*-b_1} + x^{b^*-b_2} + \cdots + x^{b^*-b_{\ell}}$, where $b^* = \max\{b_1,b_2,\ldots,b_{\ell}\}$. If $r>0$ is the initial value of the parameter, and the algorithm (a) returns a result when (or before) the parameter is negative, and (b) only executes branching rules whose roots are bounded by a constant $c >0$, then its running time is bounded by $\OO^*(c^r)$.

\subsection{Description of the Algorithm}\label{sec:describe}

Given an instance $\widehat{\I}=(\widehat{M},\widehat{W},\widehat{\cal L}_M,\widehat{\cal L}_W,\widehat{k})$ of \minBSM, we begin by using the procedure given by Lemma~\ref{lem:kernelBSM} to obtain (in polynomial time) a kernel $\I=(M,W,{\cal L}_M,{\cal L}_W,k)$ of \minBSM\ such that $\I$ has at most $3t$ men among whom at most $2t$ are \sad, at most $3t$ women among whom at most $2t$ are \sad.
Let us denote the \happy\ pairs in $\I$ by $(x_1,y_1),(x_2,y_2),\ldots,(x_h,y_h)$ for the appropriate $h\leq t$, and the set of \sad\ men by $M_S$. Note that $|M_S|\leq 2t$.

We proceed by executing a loop where each iteration corresponds to a different subset $M'\subseteq M_S$. For a specific iteration, our goal is to determine whether there exists a stable matching $\mu$ such that the following conditions are satisfied:
\begin{itemize}
\item $\bal{\mu}\leq k$.
\item For all $m\in M'$: $\mu(m)\neq \Mopt(m)$.
\item For all $m\in M_S\setminus M'$: $\mu(m)=\Mopt(m)$.
\end{itemize}

A stable matching satisfying the conditions above (in the context of the current iteration) is said to be {\em valid}. We denote $r=k-\Om$, and observe that $r\leq t$. 


 Let us now consider some specific iteration. To determine whether there exists a valid stable matching, our plan is to execute a branching procedure, called {\ttfamily Branch}, which outputs every set $S$ of pairs of a man in $M'$ and a woman, such that the following conditions are satisfied.
\begin{enumerate}
\item\label{item1} Every man $m$ in $M'$ participates in exactly one pair $(m,w)$ of $S$, and for that unique pair, it holds that $w\in{\cal A}(m)$ and $p_{m}(w)>p_m(\Mopt(m))$.
\item\label{item2} $\displaystyle{\sum_{m\in M'}(p_m(\mu(m))-p_m(\Mopt(m)))\leq r}$.
\end{enumerate}

The description of this procedure is given in the following subsection. Here, let us argue that by having this procedure, we can conclude the proof of the correctness of the algorithm. In the current iteration, we examine each set $S$ in the outputted family of sets. Then, we check whether the pairs in $S$, together with $(x_1,y_1),(x_2,y_2),\ldots,(x_h,y_h)$ and every pair in $\{(m,\Mopt(m)): m\in  M'\}$ form a stable matching whose balance is at most $k$. If the answer is positive, then we terminate the execution and accept, which is clearly correct.

At the end, if we did not accept in any iteration, we reject. To see why this decision is correct, suppose that there exists a stable matching $\mu$ whose balance is at most $k$. In this case, due to our exhaustive search, there exists some iteration in which $\mu$ is also valid. In that iteration, associated with some $M'\subseteq M-S$, observe that the set of pairs $\{(m,\mu(m)): m\in  M'\}$ is one of the outputted sets of pairs. Indeed, the satisfaction of Condition \ref{item1} follows from the fact that $\mu$ is a stable matching satisfying the last two conditions of validity, and Condition \ref{item2} follows from the fact that $\mu$ satisfies the first condition of validity.

Let us denote by $T$ the running of the procedure {\ttfamily Branch}. Then, the total running time of our algorithm is bounded by $\OO^*(2^{|M_S|}\cdot T)=\OO^*(4^t\cdot T)$, and therefore to derive the running time in Theorem \ref{thm:alg}, we will need to ensure that $T=\OO^*(2^t)$. For this purpose, it is sufficient to ensure that $T=\OO^*(2^r)$.

\subsection{The Branching Procedure}

We now present the description of the procedure {\ttfamily Branch} in the context of some set $M'\subseteq M_S$. Let us denote $M'=\{m_1,m_2,\ldots,m_p\}$ for the appropriate $p$.

 Each call to our procedure is of the form {\ttfamily Branch}$(i,S)$ where $i\in\{0,1,\ldots,p+1\}$ and $S$ is a set of pairs of a man in $\{m_1,m_2,\ldots,m_i\}$ and a woman, such that the following conditions are satisfied.
\begin{enumerate}
\item\label{item11} Every man $m$ in $\{m_1,m_2,\ldots,m_i\}$ participates in exactly one pair $(m,w)$ of $S$, and for that unique pair, it holds that $w\in{\cal A}(m)$ and $p_{m}(w)>p_m(\Mopt(m))$. We define $\mu_S$ as the {\em function} whose domain is $\{m_1,m_2,\ldots,m_i\}$ and which assigns to each man $m$ in its domain the unique woman $w$ such that $(m,w)\in S$. 
\item\label{item22} Define $\bal{S}=\displaystyle{\sum_{m\in \{m_1,m_2,\ldots,m_i\}}(p_m(\mu_S(m))-p_m(\Mopt(m)))}$. Then, $\bal{S}\leq r$.
\end{enumerate}

Note that in case $i=0$, we have that $S=\emptyset$. The call to {\ttfamily Branch} that is performed by the algorithm given in Section \ref{sec:describe} is precisely with these arguments, that is, {\ttfamily Branch}$(0,\emptyset)$. The objective of a call {\ttfamily Branch}$(i,S)$ is to return a family ${\cal F}$ of sets, where each set $F\in{\cal F}$ is a set of pairs of a man in $\{m_{i+1},m_{i+2},\ldots,m_p\}$ and a woman, such that the following conditions are satisfied.
\begin{enumerate}
\item\label{item1a} Every man $m$ in $\{m_{i+1},m_{i+2},\ldots,m_p\}$ participates in exactly one pair $(m,w)$ of $F$, and for that unique pair, it holds that $w\in{\cal A}(m)$ and $p_{m}(w)>p_m(\Mopt(m))$. We define $\mu_F$ as the {\em function} whose domain is $\{m_{i+1},m_{i+2},\ldots,m_p\}$ and which assigns to each man $m$ in its domain the unique woman $w$ such that $(m,w)\in F$. 
\item\label{item2a} $\bal{S}+\displaystyle{\sum_{m\in \{m_{i+1},m_{i+2},\ldots,m_p\}}(p_m(\mu_F(m))-p_m(\Mopt(m)))\leq r}$.
\end{enumerate}

Clearly, by accomplishing this goal, we also achieve the objective posed in Section \ref{sec:describe}. The measure that we employ to analyze our procedure is $(r-\bal{S})$ (recall that $\bal{S}$ is defined in Condition \ref{item22} of the specification of a call), which is initially equal to $r$. Hence, according to the method of bounded search trees (see Section \ref{sec:bounded}), to derive the running time $\OO^*(2^r)$, it is sufficient to ensure that {\ttfamily Branch} {\bf (a)} returns a result when (or before) the measure $r$ is negative, and {\bf (b)} only executes branching rules whose roots are bounded by 2. When the measure $r$ is negative, we should simply return ${\cal F}=\emptyset$, as there does not exist a set $F$ satisfying the conditions above. Otherwise, when $i=p+1$, we can simply return ${\cal F}=\{\emptyset\}$.

Let us now consider a call {\ttfamily Branch}$(i,S)$ where $r\geq 0$ and $i\leq p$. Denote $\widetilde{W}=\{w\in{\cal A}(m_{i+1}): p_{m_{i+1}}(w)>p_{m_{i+1}}(\Mopt({m_{i+1}}))\}$. We further refine $\widetilde{W}$ be letting $W^*$ denote the set of those $r$ women in $\widetilde{W}$ who are the most preferred by $m_{i+1}$. In case there are no such $r$ women since $|\widetilde{W}|<r$, we simply denote $W^*=\widetilde{W}$. Let us also denote $W^*=\{w_1,w_2,\ldots,w_q\}$ for the appropriate $q\leq r$. Then, our procedure executes $r$ branches. At the $j^{th}$ branch, {\ttfamily Branch} calls itself recursively with $(i+1,S\cup\{(m_{i+1},w_j)\})$. Eventually, {\ttfamily Branch} returns $\bigcup_{j=1}^q\{\{(m_{i+1},w_j)\}\cup F_j: F\in{\cal F}_j\}$ where for all $j\in\{1,2\ldots,q\}$, we set ${\cal F}_j$ to be the family of sets of pairs that was returned by the recursive call of the $j^{th}$ branch.
The correctness follows from the observation that the process that we execute is an exhaustive search. More precisely, if there exists a set $F$ satisfying Conditions \ref{item1a} and \ref{item2a}, then it must include exactly one of the pairs in $\{(m_{i+1},w_j): j\in\{1,2,\ldots,q\}\}$. Now, let us observe that at the $j^{th}$ branch, the measure changes from $r-\bal{S}$ to $r-\bal{S\cup\{({m_{i+1}},w_j)\}}$. By our definition of $w_j$, we have that $p_{m_{i+1}}(w_j)-p_{m_{i+1}}(\Mopt({m_{i+1}})))=j$. Hence, at the worst case, the branching vector is $(1,2,\ldots,r)$.
Since the root of such a branching vector upper bounded by $2$, our proof is complete.

\section{Hardness}\label{sec:hardness}

In this section, we prove the following theorem. 

\begin{theorem}\label{thm:hardness}
\maxBSM\ is \WOH.
\end{theorem}

For this purpose, we consider the {\sc Clique} problem, which is defined as follows.

\defparproblemMy{{\sc Clique}}{A graph $G=(V,E)$, and a positive integer $k$.}{$k$.}{Does $G$ contain a clique on $k$ vertices?}

The {\sc Clique} problem is known to be \WOH\ \cite{clique}. Thus, to prove Theorem \ref{thm:hardness}, it is sufficient to prove the following result.

\begin{lemma}\label{lemma:hardness}
Given an instance ${\cal I}=(G=(V,E),k)$ of {\sc Clique}, an equivalent instance $\widehat{\cal I}=(M,W,{\cal L}_M,{\cal L}_W,\widehat{k})$ of \maxBSM\ such that $t=6(k+\displaystyle{\frac{k(k-1)}{2}})$ can be constructed in time $\OO(f(k)\cdot|{\cal I}|^{\OO(1)})$ for some function $f$.
\end{lemma} 

The rest of this section focuses on the proof of Lemma~\ref{lemma:hardness}. To this end, we let ${\cal I}=(G=(V,E),k)$ be some instance of {\sc Clique}. In Section \ref{sec:reduction}, we construct (in ``\FPT\ time'') an instance $\widehat{\cal I}=(M,W,{\cal L}_M,{\cal L}_W,\widehat{k})$ of \maxBSM. This section also contains an informal explanation of the intuition underlying the construction. Then, in Section \ref{sec:t}, we verify that the parameter $t$ associated with $\widehat{\cal I}$ is equal to $6(k+\displaystyle{\frac{k(k-1)}{2}})$. Finally, in Section \ref{sec:correctness}, we prove that the input instance $\cal I$ of {\sc Clique} and our instance $\widehat{\cal I}$ of \maxBSM\ are equivalent. In what follows, we select arbitrary orders on $V$ and $E$, according to which we denote $V=\{v_1,v_2,\ldots,v_{|V|}\}$ and $E=\{e_1,e_2,\ldots,e_{|E|}\}$.

\subsection{Reduction}\label{sec:reduction}

First, to construct the sets $M$ and $W$, we define three pairwise-disjoint subsets of $M$, called $M_V, M_E$ and $\widetilde{M}$, and three pairwise-disjoint subsets of $W$, called $W_V, W_E$ and $\widetilde{W}$. Then, we set $M=M_V\cup M_E\cup \widetilde{M}\cup\{m^*\}$ and $W=W_V\cup W_E\cup \widetilde{W}\cup\{w^*\}$, where $m^*$ and $w^*$ denote a new man and a new woman (which do not belong to the six sets defined previously), respectively. 

\begin{itemize}
\item $M_V=\{m^i_v: v\in V, i\in\{1,2\}\}$; $W_V=\{w^i_v: v\in V, i\in\{1,2\}\}$.
\medskip
\item $M_E=\{m^i_e: e\in E, i\in\{1,2\}\}$; $W_E=\{w^i_e: e\in E, i\in\{1,2\}\}$.
\medskip
\item Let $\delta=2(|V|+|E|+|V||E|+|V||E|^2)-k(4+4k+2|E|+(k-1)|V||E|)$.\\
Then, $\widetilde{M}=\{\widetilde{m}^i: i\in\{1,2,\ldots,\delta\}\}$ and $\widetilde{W}=\{\widetilde{w}^i: i\in\{1,2,\ldots,\delta\}\}$.
\end{itemize}

Note that $|M|=|W|$. We remark that in what follows, we assume w.l.o.g.~that $\delta\geq 0$ and $|V|>k+\displaystyle{\frac{k(k-1)}{2}}$, else the size of the input instance $\cal I$ of {\sc Clique} is bounded by a function of $k$ and can therefore, by using brute-force, be solved in FPT time.

Before we proceed, let us discuss the intuition behind the definition of the subsets of men and women above, and the definitions of the preference lists that will follow. Roughly speaking, each pair of men, $m^1_v$ and $m^2_v$, represents a vertex, and we aim to ensure that either both men will be matched to their best partners (in the man-optimal stable matching) or both men will be matched to other partners (where there would be only {\em one} choice for these other partners that preserves stability). Accordingly, we will guarantee that the choice of matching these two men to their best partners translates to {\em not} choosing the vertex they represent into the clique, and the other choice translates to choosing this vertex into the clique.

Now, having just the set $M_V$, we can encode selection of vertices into the clique, but we cannot ensure that the vertices we select indeed form a clique. For this purpose, we also have the set $M_E$ which, in a manner similar to $M_V$, encodes selection of edges into the clique. By designing the instance in a way that the situation of the men in the man-optimal stable matching is significantly worse than that of the women in the women-optimal stable matching, we are able to ensure that at most $\displaystyle{2(k+\frac{k(k-1)}{2})}$ men in $M_V\cup M_E$ will not be assigned their best partners (here, we exploit the condition that $\bal{\mu}\leq \widehat{k}$ for a solution $\mu$). We remark that here the man $m^*$ plays a crucial role---by using dummy men and women (in the sets $\widetilde{M}$ and $\widetilde{W}$) that prefer each other over all other people, we ensure that the situation of $m^*$ is always ``extremely bad'' (from his viewpoint), while the situation of his partner, $w^*$, is always ``excellent'' (from her viewpoint).

At this point, we first need to ensure that the edges that we select indeed connect the vertices that we select. For this purpose, we carefully design our reduction so that when a pair of men representing some edge $e$ obtain partners worse than those they have in the man-optimal stable matching, it must be that the men representing the endpoints of $e$ have also obtained partners worse than those they have in the man-optimal stable matching, else stability will not be preserved---the partners of the men represting the endpoints of $e$ will form blocking pairs together with the men representing~$e$.

Finally we observe that we still need to ensure that among our $\displaystyle{2(k+\frac{k(k-1)}{2})}$ distinguished men in $M_V\cup M_E$, which are associated with $\displaystyle{k+\frac{k(k-1)}{2}}$ selected elements (vertices and edges), there will be exactly $2k$ distinguished men from $M_V$ and exactly $k(k-1)$ distinguished men from $M_E$, which would mean we have chosen $k$ vertices and $\displaystyle{\frac{k(k-1)}{2}}$ edges. For this purpose, we construct an instance where for the women, it is only somewhat ``beneficial'' that the men in $M_V$ will not be matched to their best partners, but it is extremely beneficial that the men in $M_E$ will not be matched to their best partners. This objective is achieved by carefully placing dummy men (from $\widehat{M}$) in the preference lists of women in $W_E$. By again exploiting the condition that $\bal{\mu}\leq \widehat{k}$ for a solution $\mu$, we are able to ensure that there would be at least $k(k-1)$ distinguished men from~$M_E$.

Next, we proceed with the formal presentation of our reduction by defining $p_m$ for every man $m\in M$, and thus constructing ${\cal L}_M$.

\begin{itemize}
\item For all $m^1_v\in M_V$: ${\cal A}(m^1_v)=\{w^1_v,\widetilde{w}^1,\widetilde{w}^2,w^2_v\}$.
	\begin{itemize}
	\item $p_{m^1_v}(w^1_v)=1$; $p_{m^1_v}(\widetilde{w}^1)=2$; $p_{m^1_v}(\widetilde{w}^2)=3$; $p_{m^1_v}(w^2_v)=4$.
	\end{itemize}
\medskip
\item For all $m^2_v\in M_V$: ${\cal A}(m^2_v)=\{w^2_v,w^1,w^2,w^1_v\}$.
	\begin{itemize}
	\item $p_{m^2_v}(w^2_v)=1$; $p_{m^2_v}(\widetilde{w}^1)=2$; $p_{m^2_v}(\widetilde{w}^2)=3$; $p_{m^2_v}(w^1_v)=4$.
	\end{itemize}
\medskip
\item For all $m^1_{\{u,v\}}\in M_E$ where $u<v$:\footnote{Recall that we have defined an order on $V$.} ${\cal A}(m^1_{\{u,v\}})=\{w^1_{\{u,v\}},w^1_u,w^1_v,w^2_{\{u,v\}}\}$.
	\begin{itemize}
	\item $p_{m^1_{\{u,v\}}}(w^1_{\{u,v\}})=1$; $p_{m^1_{\{u,v\}}}(w^1_u)=2$; $p_{m^1_{\{u,v\}}}(w^1_v)=3$; $p_{m^1_{\{u,v\}}}(w^2_{\{u,v\}})=4$.
	\end{itemize}
\medskip
\item For all $m^2_{\{u,v\}}\in M_E$ where $u<v$: ${\cal A}(m^2_{\{u,v\}})=\{w^2_{\{u,v\}},w^2_u,w^2_v,w^1_{\{u,v\}}\}$.
	\begin{itemize}
	\item $p_{m^2_{\{u,v\}}}(w^2_{\{u,v\}})=1$; $p_{m^2_{\{u,v\}}}(w^2_u)=2$; $p_{m^2_{\{u,v\}}}(w^2_v)=3$; $p_{m^2_{\{u,v\}}}(w^1_{\{u,v\}})=4$.
	\end{itemize}
\medskip
\item For all $\widetilde{m}^i\in \widetilde{M}$ such that $i\leq |V||E|$: ${\cal A}(\widetilde{m}^i)=\{\widetilde{w}^i\}\cup W_E\cup\{w^i_v\in M_V: \widetilde{m}^i\in\C{A}(w^i_v)\}$, where for all $w^i_v\in W_V$, the set $\C{A}(w^i_v)$ is determined later.
	\begin{itemize}
	\item $p_{\widetilde{m}^i}(\widetilde{w}^i)=1$.
	\item For all $j\in\{1,2,\ldots,|E|\}$: $p_{\widetilde{m}^i}(w^1_{e_j})=j+1$; $p_{\widetilde{m}^i}(w^2_{e_j})=|E|+j+1$.
	\item Denote $X=\{w^i_v\in M_V: \widetilde{m}^i\in\C{A}(w^i_v)\}$. Let $f: X\rightarrow |X|$ be an arbitrarily chosen bijection. Then, for all $w^i_v\in X$: $p_{\widetilde{m}^i}(w^i_v)=2|E|+f(w^i_v)+1$.
	\end{itemize} 
\item For all $\widetilde{m}^i\in \widetilde{M}$ such that $i>|V||E|$: ${\cal A}(\widetilde{m}^i)=\{\widetilde{w}^i\}$.
		\begin{itemize}
		\item $p_{\widetilde{m}^i}(\widetilde{w}^i)=1$.
		\end{itemize}	
\medskip
\item For $m^*$: ${\cal A}(m^*)=\widetilde{W}\cup\{w^*\}$.
	\begin{itemize}
	\item For all $i\in\{1,2,\ldots,\delta\}$: $p_{m^*}(\widetilde{w}^i)=i$.
	\item $p_{m^*}(w^*)=\delta+1$.
	\end{itemize}
\end{itemize}

Accordingly, we define $p_w$ for every woman $w\in W$, and thus construct ${\cal L}_W$.

\begin{itemize}
\item For all $w^1_v\in W_V$: ${\cal A}(w^1_v)=\{m^1_e\in M_E: v\in e\}\cup\{m^2_v,m^1_v\}\cup\{\widetilde{m}^i\in\widetilde{M}: i\leq |E|-\mathrm{degree}_G(v)\}$.
	\begin{itemize}
	\item $p_{w^1_v}(m^2_v)=1$
	\item For all $i\in\{1,2,\ldots,|E|\}$: $p_{w^1_v}(m^1_{e_i})=i+1$.
	\item $p_{w^1_v}(m^1_v)=|E|+2$.
	\end{itemize}
\medskip
\item For all $w^2_v\in W_V$: ${\cal A}(w^2_v)=\{m^2_e\in M_E: v\in e\}\cup\{m^1_v,m^2_v\}\cup\{\widetilde{m}^i\in\widetilde{M}: i\leq |E|-\mathrm{degree}_G(v)\}$.
	\begin{itemize}
	\item $p_{w^2_v}(m^1_v)=1$.
	\item For all $i\in\{1,2,\ldots,|E|\}$: $p_{w^2_v}(m^2_{e_i})=i+1$.
	\item $p_{w^2_v}(m^2_v)=|E|+2$.
	\end{itemize}
\medskip
\item For all $w^1_e\in W_E$: ${\cal A}(w^1_e)=\{\widetilde{m}^i\in\widetilde{M}: i\in\{1,2,\ldots,|V||E|\}\}\cup\{m^2_e,m^1_e\}$.
	\begin{itemize}
	\item $p_{w^1_e}(m^2_e)=1$.
	\item For all $i\in\{1,2,\ldots,|V||E|\}$: $p_{w^1_e}(\widetilde{m}^i)=i+1$.
	\item $p_{w^1_e}(m^1_e)=|V||E|+2$.
	\end{itemize}
\medskip
\item For all $w^2_e\in W_E$: ${\cal A}(w^2_e)=\{\widetilde{m}^i\in\widetilde{M}: i\in\{1,2,\ldots,|V||E|\}\}\cup\{m^1_e,m^2_e\}$.
	\begin{itemize}
	\item $p_{w^2_e}(m^1_e)=1$.
	\item For all $i\in\{1,2,\ldots,|V||E|\}$: $p_{w^2_e}(\widetilde{m}^i)=i+1$.
	\item $p_{w^2_e}(m^2_e)=|V||E|+2$.
	\end{itemize}
\medskip
\item For all $\widetilde{w}^i\in \widetilde{W}$ such that $i\in\{1,2\}$: ${\cal A}(\widetilde{w}^i)=\{\widetilde{m}^i\}\cup\{m^*\}\cup M_V$.
	\begin{itemize}
		\item $p_{\widetilde{w}^i}(\widetilde{m}^i)=1$; $p_{\widetilde{w}^i}(m^*)=2$.
		\item For all $j\in\{1,2,\ldots,|V|\}$: $p_{\widetilde{w}^i}(m^1_{v_j})=j+2$; $p_{\widetilde{w}^i}(m^2_{v_j})=|V|+j+2$.
		\end{itemize}
	\item For all $\widetilde{w}^i\in \widetilde{W}$ such that $i>2$: ${\cal A}(\widetilde{w}^i)=\{\widetilde{m}^i\}\cup\{m^*\}$.
		\begin{itemize}
		\item $p_{\widetilde{w}^i}(\widetilde{m}^i)=1$; $p_{\widetilde{w}^i}(m^*)=2$.
		\end{itemize}
\medskip
\item For $w^*$: ${\cal A}(w^*)=\{m^*\}$.
	\begin{itemize}
	\item $p_{w^*}(m^*)=1$.
	\end{itemize}
\end{itemize}

Finally, we define $\widehat{k}=|M|+\delta+6(k+\displaystyle{\frac{k(k-1)}{2}})$. It is clear that the entire construction (under the assumptions that $\delta\geq 0$ and $|V|>k+\displaystyle{\frac{k(k-1)}{2}}$) can be performed in polynomial time.

\subsection{The Parameter}\label{sec:t}

Our current objective is to verify that $t$ is indeed bounded by a function of $k$. For this purpose, we first observe that for all $i\in\{1,2,\ldots,\delta\}$, it holds that $p_{\widetilde{m}^i}(\widetilde{w}^i)=p_{\widetilde{w}^i}(\widetilde{m}^i)=1$. Therefore, for all $\mu\in\mathrm{SM}(\widehat{\cal I})$ and $i\in\{1,2,\ldots,\delta\}$, we have that $\mu(\widetilde{m}^i)=\widetilde{w}^i$, else $(\widetilde{m}^i,\widetilde{w}^i)$ would have formed a blocking pair for $\mu$ in~$\widehat{\cal I}$.

\begin{observation}\label{obs:dummy1}
For all $\mu\in SM(\widehat{\cal I})$ and $i\in\{1,2,\ldots,\delta\}$, it holds that $\mu(\widetilde{m}^i)=\widetilde{w}^i$.
\end{observation}

Now, note that ${\cal A}(m^*)=\widetilde{W}\cup\{w^*\}$. Thus, by Observation \ref{obs:dummy1}, we have that for all $\mu\in SM(\widehat{\cal I})$, either $m^*$ is unmatched or $\mu(m^*)=w^*$. However, ${\cal A}(w^*)=\{m^*\}$, which implies that in the former case, $(m^*,w^*)$ forms a blocking pair. Thus, we also have the following observation.

\begin{observation}\label{obs:dummy2}
For all $\mu\in SM(\widehat{\cal I})$, it holds that $\mu(m^*)=w^*$.
\end{observation}

Let us proceed by identifying the man-optimal $\mu_M$ and the woman-optimal $\mu_W$ stable matchings. For this purpose, we first define a matching $\mu'_M$ as follows.
\begin{itemize}
\setlength{\itemsep}{-1pt}
\item For all $m^i_v\in M_V$: $\mu'_M(m^i_v)=w^i_v$.
\item For all $m^i_e\in M_E$: $\mu'_M(m^i_e)=w^i_e$.
\item For all $\widetilde{m}^i\in \widetilde{M}$: $\mu'_M(\widetilde{m}^i)=\widetilde{w}^i$.
\item $\mu'_M(m^*)=w^*$.
\end{itemize}

\begin{lemma}\label{lem:manOpt}
$\mu_M=\mu'_M$.
\end{lemma}

\begin{proof}
Since for all $i\in\{1,2,\ldots,\delta\}$, it holds that $\mu'_M(\widetilde{m}^i)=\widetilde{w}^i$ and $p_{\widetilde{m}^i}(\widetilde{w}^i)=p_{\widetilde{w}^i}(\widetilde{m}^i)=1$, we have that there cannot exist a blocking pair with at least one person from $\widetilde{M}\cup\widetilde{W}$. Now, notice that for every $m\in M$, including $m^*$, the woman most preferred by $m$ who is outside $\widetilde{W}$ is also the one with whom it is matched. Therefore, there cannot exist any blocking pair for $\mu'_M$, and by Observation \ref{obs:dummy1}, we further conclude that indeed $\mu_M=\mu'_M$.
\end{proof}

Now, we define a matching $\mu'_W$ as follows.

\begin{itemize}
\setlength{\itemsep}{-1pt}
\item For all $w^i_v\in W_V$: $\mu'_W(w^i_v)=m^{3-i}_v$.
\item For all $w^i_e\in W_E$: $\mu'_W(w^i_e)=m^{3-i}_e$.
\item For all $\widetilde{w}^i\in \widetilde{W}$: $\mu'_W(\widetilde{w}^i)=\widetilde{m}^i$.
\item $\mu'_W(w^*)=m^*$.
\end{itemize}

\begin{lemma}\label{lem:womanOpt}
$\mu_W=\mu'_W$.
\end{lemma}

\begin{proof}
In the matching $\mu'_W$, every woman is matched with the man she prefers the most. Thus, it is immediate that $\mu_W=\mu'_W$.
\end{proof}

As a corollary to Lemmata \ref{lem:manOpt} and \ref{lem:womanOpt}, we obtain the following result.

\begin{corollary}\label{cor:OmOw}
$\Om=|M|+\delta$ and $\Ow=|W|$.
\end{corollary}

\begin{proof}
First, note that
\[\begin{array}{ll}
O_M & = \displaystyle{\sum_{(m,w)\in \mu_M} p_m(w)}\\
    & = p_{m^*}(\mu_M(m^*)) + \displaystyle{\sum_{m\in M\setminus\{m^*\}} p_m(\mu_M(m))} = (\delta+1)+(|M|-1) = |M|+\delta.
\end{array}\]

Second, note that
\[\begin{array}{ll}
O_W & = \displaystyle{\sum_{(m,w)\in \mu_W} p_w(m)} = \displaystyle{\sum_{w\in W} p_w(\mu_W(w))} = |W|.
\end{array}\]
\end{proof}

We are now ready to bound $t$.

\begin{lemma}\label{lem:t}
The parameter $t$ associated with $\widehat{\cal I}$ is equal to $6(k+\displaystyle{\frac{k(k-1)}{2}})$.
\end{lemma}

\begin{proof}
By the definition of $t$, we have that
\[\begin{array}{ll}
t & = \widehat{k}-\max\{\Om,\Ow\}\\
  & = |M|+\delta+6(k+\displaystyle{\frac{k(k-1)}{2}}) - \max\{|M|+\delta, |W|\}= 6(k+\displaystyle{\frac{k(k-1)}{2}}).
\end{array}\]
\end{proof}

\subsection{Correctness}\label{sec:correctness}

First, from Lemma \ref{lem:manOpt} and Proposition \ref{lem:rht} we derive the following useful observation.

\begin{observation}\label{obs:matchAll}
Every $\mu\in SM(\widehat{\cal I})$ matches all people in $M\cup W$.
\end{observation}

Next, we proceed to state the first direction necessary to conclude that the input instance $\cal I$ of {\sc Clique} and our instance $\widehat{\cal I}$ of \maxBSM\ are equivalent.

\begin{lemma}\label{lem:hardFirst}
If $\cal I$ is a \Yes-instance, then $\widehat{\cal I}$ is a \Yes-instance.
\end{lemma}

\begin{proof}Suppose that $\cal I$ is a \Yes-instance, and let $U$ be the vertex-set of a clique on $k$ vertices in $G$. We denote $M^U_V=\{m^i_v\in M_V: v\in U\}$ and $M^U_E=\{m^i_{\{u,v\}}\in M_E: u,v\in U\}$. Then, we define a matching $\mu$ as follows.
\begin{itemize}
\item For all $m^i_{v}\in M_V$:
	\begin{itemize}
	\item If $m^i_{v}\in M^U_V$: $\mu(m^i_v)=w^{3-i}_v$.
	\item Else: $\mu(m^i_v)=w^i_v$.
	\end{itemize}
\item For all $m^i_e\in M_E$:
	\begin{itemize}
	\item If $m^i_e\in M^U_E$: $\mu(m^i_e)=w^{3-i}_e$.
	\item Else: $\mu(m^i_e)=w^i_e$.
	\end{itemize}
\item For all $\widetilde{m}^i\in\widetilde{M}$: $\mu(\widetilde{m}^i)=\widetilde{w}^i$.
\item $\mu(m^*)=w^*$.
\end{itemize}

We claim that $\mu\in SM(\widehat{\cal I})$ and $\bal{\mu}\leq \widehat{k}$, which would imply that $\widehat{\cal I}$ is a \Yes-instance. To this end, we first show that $\mu\in SM(\widehat{\cal I})$. Since for all $i\in\{1,2,\ldots,\delta\}$, it holds that $\mu(\widetilde{m}^i)=\widetilde{w}^i$ and $p_{\widetilde{m}^i}(\widetilde{w}^i)=p_{\widetilde{w}^i}(\widetilde{m}^i)=1$, we have that there cannot exist a blocking pair with at least one person from $\widetilde{M}\cup\widetilde{W}$. Thus, there can also not be a blocking pair with any person from $\{m^*,w^*\}$.

On the one hand, notice that for every $m\in (M_V\setminus M^U_V)\cup(M_E\setminus M^U_E)\cup\{m^*\}$, the woman most preferred by $m$ who is outside $\widetilde{W}$ is also the one with whom it is matched. Thus, no man in $(M_V\setminus M^U_V)\cup(M_E\setminus M^U_E)\cup\{m^*\}$ can belong to a blocking pair. Moreover, the set of acceptable partners of any woman in $W_E$ matched to a man in $M_E\setminus M^U_E$ is a subset of $\widetilde{M}\cup (M_E\setminus M^U_E)$, and therefore such a woman cannot belong to a blocking pair. On the other hand, let $W'$ denote the set of every woman that is matched to a man $m\in M^U_V\cup M^U_E$. Then, for every $w\in W'$, the man most preferred by $w$ is also the one with whom she is matched. Therefore, no woman in $W'$ can belong to a blocking pair. Hence, we also conclude that no woman in $W_E$ can belong to a blocking pair.

Thus, if there exists a blocking pair, it must consist of a man $m\in M^U_V\cup M^U_E$ and a woman $w\in W_V\setminus W'$. Suppose, by way of contradiction, that there exists such a blocking pair $(m,w)$. First, let us assume that $m=m^i_v\in M^U_V$. In this case, since apart from $w^i_v$, all women in ${\cal A}(m^i_v)$ belong to $\widetilde{W}\cup\{\mu(m^i_v)\}$, we deduce that $w=w^i_v$. However, $w^i_v$ prefers $\mu(w^i_v)$ over $m^i_v$, and thus we reach a contradiction. Next, we assume that $m=m^i_{\{u,v\}}\in M^U_E$. In this case, it must hold that $w$ is either $w^i_v$ or $w^i_u$. Without loss of generality, we assume that $w=w^i_v$. However, since $m^i_{\{u,v\}}\in M^U_E$, we have that $v\in U$. Therefore, $\mu(w^i_v)=m^{3-i}_v$. Since $w^i_v$ prefers $m^{3-i}_v$ over $m^i_{\{u,v\}}$, we reach a contradiction.

It remains to prove that $\bal{\mu}\leq\widehat{k}$. To this end, we need to show that 
\[\displaystyle{\max\{\sum_{(m,w) \in \mu} p_m(w), \sum_{(m,w) \in \mu} p_w(m)\}\leq |M|+\delta+6(k+\displaystyle{\frac{k(k-1)}{2}})}.\]

First, note that
\[\begin{array}{lll}
\displaystyle{\sum_{(m,w) \in \mu} p_m(w)} & = &  \displaystyle{\sum_{m\in M^U_V} p_m(\mu(m)) + \sum_{m\in M_V\setminus M^U_V} p_m(\mu(m)) + \sum_{m\in M^U_E} p_m(\mu(m))}\\

  && + \displaystyle{\sum_{m\in M_E\setminus M^U_E} p_m(\mu(m)) + \sum_{m\in\widetilde{M}}p_m(\mu(m)) + p_{m^*}(\mu(m^*))}\\
	
  & = &  4|M^U_V| + |M_V\setminus M^U_V| + 4|M^U_E| + |M_E\setminus M^U_E| + |\widetilde{M}| + \delta+1\\
	
	& = &  |M| + \delta + 3(|M^U_V|+|M^U_E|) = |M| + \delta + 6(k+\displaystyle{\frac{k(k-1)}{2}}).
\end{array}\]

Second, note that

\[\begin{array}{lll}
\displaystyle{\sum_{(m,w) \in \mu} p_w(m)} & = &  \displaystyle{\sum_{m\in M^U_V} p_{\mu(m)}(m) + \sum_{m\in M_V\setminus M^U_V} p_{\mu(m)}(m) + \sum_{m\in M^U_E} p_{\mu(m)}(m)}\\

  && + \displaystyle{\sum_{m\in M_E\setminus M^U_E} p_{\mu(m)}(m) + \sum_{m\in\widetilde{M}}p_{\mu(m)}(m) + p_{{\mu(m^*)}}(m^*)}\\
	
  & = &  |M^U_V| + |M_V\setminus M^U_V|(|E|+2) + |M^U_E| + |M_E\setminus M^U_E|(|V||E|+2) + |\widetilde{M}| + 1 \\
	
  & = &  |M| + 2(|V|-k)(|E|+1) + 2(|E|-\displaystyle{\frac{k(k-1)}{2}})(|V||E|+1) \\	
	
	
	& = & |M| + \delta + 6(k+\displaystyle{\frac{k(k-1)}{2}}).
\end{array}\]

This concludes the proof of the lemma.
\end{proof}

We now turn to prove the second direction.

\begin{lemma}\label{lem:hardSecond}
If $\widehat{\cal I}$ is a \Yes-instance, then $\cal I$ is a \Yes-instance.
\end{lemma}

\begin{proof}
Suppose that $\widehat{\cal I}$ is a \Yes-instance, and let $\mu$ be a stable matching such that $\bal{\mu}\leq\widehat{k}$. By Observations \ref{obs:dummy1} and \ref{obs:dummy2}, it holds that
\begin{itemize}
\item For all $i\in\{1,2,\ldots,\delta\}$: $\mu(\widetilde{m}^i)=\widetilde{w}^i$.
\item $\mu(m^*)=w^*$.
\end{itemize}

Thus, since Observation \ref{obs:matchAll} implies that all vertices in $M_V$ should be matched by $\mu$, we deduce that
\begin{itemize}
\item For all $v\in V$: Either both $\mu(m^1_v)=w^1_v$ and $\mu(m^2_v)=w^2_v$ or both $\mu(m^2_v)=w^1_v$ and $\mu(m^1_v)=w^2_v$.
\end{itemize}

Let $U$ denote the set of every $v\in V$ such that $\mu(m^2_v)=w^1_v$ and $\mu(m^1_v)=w^2_v$. Moreover, denote $M^U_V=\{m^i_v\in M_V: v\in U\}$. By the item above, and since all vertices in $M_E$ should also be matched by $\mu$, we further deduce that
\begin{itemize}
\item For all $e\in E$: Either both $\mu(m^1_e)=w^1_e$ and $\mu(m^2_e)=w^2_e$ or both $\mu(m^2_e)=w^1_e$ and $\mu(m^1_e)=w^2_e$.
\end{itemize}

Let $S$ denote the set of every $e\in E$ such that $\mu(m^2_e)=w^1_e$ and $\mu(m^1_e)=w^2_e$. Moreover, denote $M^S_E=\{m^i_e\in M_E: e\in S\}$. If there existed $\{u,v\}\in S$ such that $u\notin U$, then $(m^1_{\{u,v\}},w^1_u)$ would have formed a blocking pair, which contradicts the fact that $\mu$ is a stable matching. Thus, we have that the set of endpoints of the edges in $S$ is a subset of $U$.

We claim that $|U|=k$ and that $U$ is the vertex-set of a clique in $G$, which would imply that $\cal I$ is a \Yes-instance. Since we have argued that the set of endpoints of the edges in $S$ is a subset of $U$, it is sufficient to show that $|U|\leq k$ and $|S|\geq \displaystyle{\frac{k(k-1)}{2}}$ (note that $|S|\geq \displaystyle{\frac{k(k-1)}{2}}$ implies that $|U|\geq k$), as this would imply that $U$ is indeed the vertex-set of a clique on $k$ vertices in $G$. First, since $\bal{\mu}\leq \widehat{k}$, we have that $\displaystyle{\sum_{(m,w)\in\mu}p_m(w)\leq |M| + \delta + 6(k+\displaystyle{\frac{k(k-1)}{2}})}$. Now, note that
\[\begin{array}{lll}
\displaystyle{\sum_{(m,w) \in \mu} p_m(w)} & = &  \displaystyle{\sum_{m\in M^U_V} p_m(\mu(m)) + \sum_{m\in M_V\setminus M^U_V} p_m(\mu(m)) + \sum_{m\in M^S_E} p_m(\mu(m))}\\

  && + \displaystyle{\sum_{m\in M_E\setminus M^S_E} p_m(\mu(m)) + \sum_{m\in\widetilde{M}}p_m(\mu(m)) + p_{m^*}(\mu(m^*))}\\
	
  & = & 4|M^U_V| + |M_V\setminus M^U_V| + 4|M^S_E| + |M_E\setminus M^S_E| + |\widetilde{M}| + \delta+1\\	

  & = & |M| + \delta + 6(|U|+|S|).
\end{array}\]

Thus, we deduce that $|U|+|S|\leq k+\displaystyle{\frac{k(k-1)}{2}}$. Now, observe that since $\bal{\mu}\leq \widehat{k}$, we also have that $\displaystyle{\sum_{(m,w)\in\mu}p_{w}(m)\leq |M| + \delta + 6(k+\displaystyle{\frac{k(k-1)}{2}})}$. Here, on the one hand we note that
\[\begin{array}{lll}
\displaystyle{\sum_{(m,w) \in \mu} p_w(m)} & = &  \displaystyle{\sum_{m\in M^U_V} p_{\mu(m)}(m) + \sum_{m\in M_V\setminus M^U_V} p_{\mu(m)}(m) + \sum_{m\in M^S_E} p_{\mu(m)}(m)}\\

  && + \displaystyle{\sum_{m\in M_E\setminus M^S_E} p_{\mu(m)}(m) + \sum_{m\in\widetilde{M}}p_{\mu(m)}(m) + p_{{\mu(m^*)}}(m^*)}\\
	
  & = & |M^U_V| + |M_V\setminus M^U_V|(|E|+2) + |M^S_E| + |M_E\setminus M^S_E|(|V||E|+2) + |\widetilde{M}| + 1\\
	
	& = & |M| + 2(|V|-|U|)(|E|+1) + 2(|E|-|S|)(|V||E|+1)\\
	
	& = & |M| + 2(|V|+|E|+|V||E|+|V||E|^2)- 2|U|(|E|+1) - 2|S|(|V||E|+1).
\end{array}\]

On the other hand, we note that
\[\begin{array}{lll}
\widehat{k} & = &  |M| + \delta + 6(k+\displaystyle{\frac{k(k-1)}{2}})\\

  & = & |M| + 2(|V|+|E|+|V||E|+|V||E|^2) -k(4+4k+2|E|+(k-1)|V||E|) + 6(k+\displaystyle{\frac{k(k-1)}{2}})\\
	
  & = & |M| + 2(|V|+|E|+|V||E|+|V||E|^2) - \displaystyle{2k(|E|+1) - k(k-1)(|V||E|+1)}.
\end{array}\]

Thus, we have that
\[|U|(|E|+1) + |S|(|V||E|+1)\geq \displaystyle{k(|E|+1) + \frac{k(k-1)}{2}(|V||E|+1)}\]

Recall that we have also shown that  $|U|+|S|\leq \displaystyle{k+\frac{k(k-1)}{2}}$. Thus, since $|U|\leq \displaystyle{k+\frac{k(k-1)}{2}} < |V|$, to satisfy the above equation it must hold that $|S|\geq \displaystyle{\frac{k(k-1)}{2}}$. Since $|U|+|S|\leq \displaystyle{k+\frac{k(k-1)}{2}}$, we deduce that $|U|\leq k$. This, as we have argued earlier, finished the proof.
\end{proof}

This concludes the proof of Theorem \ref{thm:hardness}.

\bibliographystyle{siam}
\bibliography{References}
\end{document}